\documentclass[a4paper,10pt]{article}
\usepackage{amsmath,amsfonts,amssymb,amsthm}
\usepackage{latexsym}
\usepackage{mathrsfs}
\usepackage{color}
\usepackage{graphicx} 
\usepackage{hyperref}
\usepackage{enumerate}

\usepackage[margin=2.5cm]{geometry}
\usepackage{fancyhdr}

\definecolor{detailsgray}{gray}{0.3}

\def \cbb{\mathbb{C}}

\def \nbb{\mathbb{N}}
\def \rbb{\mathbb{R}}

\def \ccal {\mathcal{C}}

\def \fcal {\mathcal{F}}

\def \hcal {\mathcal{H}}
\def \ical {\mathcal{I}}

\def \kcal {\mathcal{K}}
\def \lcal {\mathcal{L}}
\def \mcal {\mathcal{M}}
\def \ncal {\mathcal{N}}
\def \ocal {\mathcal{O}}
\def \pcal {\mathcal{P}}
\def \qcal {\mathcal{Q}}

\def \scal {\mathcal{S}}

\def \ucal {\mathcal{U}}

\def \wcal {\mathcal{W}}

\def \bfk  {\mathfrak{B}}

\newcommand{\thetav}{\boldsymbol{\theta}}
\newcommand{\etav}{\boldsymbol{\eta}}
\newcommand{\lambdav}{\boldsymbol{\lambda}}

\newcommand{\zetav}{\boldsymbol{\zeta}}
\newcommand{\piv}{\boldsymbol{\pi}}

\newcommand{\ev}{\boldsymbol{e}}

\newcommand{\zerov}{\boldsymbol{0}}

\def \. { \,\! }

\def\clap#1{\hbox to 0pt{\hss#1\hss}}

\def \cdotarg { \, \cdot \, }

\DeclareMathOperator{\im}{Im}
\DeclareMathOperator{\re}{Re}

\newcommand{\id}{\operatorname{id}}
\newcommand{\idop}{\boldsymbol{1}}

\DeclareMathOperator{\supp}{supp}
\DeclareMathOperator{\dom}{dom}

\def \expltext#1 {\\ \text{\footnotesize{ (#1) }}\\}
\def \intercomm#1 {\\ \text{\footnotesize{ (#1) }}\\}
\def \undercomm#1 {\underset{\text{\scriptsize{ (#1) }}}}
\def \overcomm#1 {\overset{\text{\scriptsize{ (#1) }}}}

\newcommand{\Hil}{\mathcal{H}}

\newcommand{\boundedops}{\bfk(\Hil)}

\newcommand{\qf}{\mathcal{Q}}

\newcommand{\A}{\mathcal{A}}
\newcommand{\F}{\mathcal{F}}
\newcommand{\M}{\mathcal{M}}

\newcommand{\hscalar}[2]{\langle #1 , #2 \rangle }

\newcommand{\gnorm}[2]{\lVert #1 \rVert_{#2}}

\newcommand{\onorm}[2]{\lVert #1 \rVert_{#2}^\omega}

\newcommand{\zd}{z^\dagger}

\newcommand{\cmeA}[2]{f_{#1}^{\lbrack #2 \rbrack}}
\newcommand{\cmeB}[2]{f_{#1}{\lbrack #2 \rbrack}}
\newcommand{\cme}[2]{\mathchoice{\cmeA{#1}{#2}}{\cmeB{#1}{#2}}{\cmeB{#1}{#2}}{\cmeB{#1}{#2}}}

\newcommand{\perms}[1]{\mathfrak{S}_{#1}}

\DeclareMathOperator*{\res}{res}

\newcommand{\scin}{\mathrm{in}}
\newcommand{\scout}{\mathrm{out}}
\newcommand{\scinout}{\mathrm{in/out}}

\newcommand{\twist}{\mathsf{t}\mspace{1.5mu}}

\DeclareMathOperator{\adj}{ad}

\newcommand{\cinfty}{\mathcal{C}^\infty}

\newcommand{\diso}{\mathcal{D}}

\newcommand{\fock}{\mathscr{F}}
\newcommand{\mo}{W}
\newcommand{\grad}{\Gamma}
\newcommand{\wedg}{\mathcal{W}}

\newcommand{\trans}{\mathsf{t}}
\newcommand{\boost}{\mathsf{b}}
\newcommand{\refl}{\mathsf{j}}

\newcommand{\raps}[1]{\mathcal{R}\lbrack {#1} \rbrack}
\newcommand{\pfg}{\phi}

\newcommand{\email}[1]{\mbox{\href{mailto:#1}{#1}}}

\newtheorem{definition}{Definition}[section]
\newtheorem{lemma}[definition]{Lemma}
\newtheorem{proposition}[definition]{Proposition}
\newtheorem{theorem}[definition]{Theorem}
\newtheorem{corollary}[definition]{Corollary}

\numberwithin{equation}{section}

\title{Fermionic integrable models and graded Borchers triples}
\author{
Henning Bostelmann\thanks{%
University of York, Department of Mathematics, York YO10 5DD, United Kingdom. \newline
Present address: Hochschule Merseburg, Fachbereich Ingenieur- und Naturwissenschaften, Eberhard-Leibnitz-Straße 2, 06217 Merseburg, Germany; 
e-mail: \email{henning.bostelmann@hs-merseburg.de}}
\and 
Daniela Cadamuro\thanks{%
Institut f\"ur Theoretische Physik, Universit\"at Leipzig, Br\"uderstra\ss{}e 16, 04103 Leipzig, Germany;  \newline e-mail: \email{cadamuro@itp.uni-leipzig.de}}
}

\date{February 16, 2024}

\begin{document}

\maketitle

\begin{abstract}
We provide an operator-algebraic construction of integrable models of quantum field theory on 1+1 dimensional Minkowski space with fermionic scattering states. These are obtained by a grading of the wedge-local fields or, alternatively, of the underlying Borchers triple defining the theory. This leads to a net of graded-local field algebras, of which the even part can be considered observable, although it is lacking Haag duality. Importantly, the nuclearity condition implying nontriviality of the local field algebras is independent of the grading, so that existing results on this technical question can be utilized. Application of Haag-Ruelle scattering theory confirms that the asymptotic particles are indeed fermionic. We also discuss connections with the form factor programme.
\end{abstract}

\tableofcontents

\section{Introduction}

The task of rigorously constructing interacting models remains one of the central questions of quantum field theory. In the past two decades, progress has been made in low-dimensional models, in particular in the class of integrable quantum field theories (where the scattering matrix is factorizing). 

A substantial class of these integrable models are now amenable to a construction with operator-algebraic techniques, directly on Minkowski space and without the need for a Euclidean formulation, beginning with the scalar models constructed by Lechner \cite{Lechner:2008} which have then progressively been generalized, e.g., in \cite{LechnerSchuetzenhofer:2012,AL2017,Tanimoto:federbush,CT2015}.

These models are essentially built in an ``inverse scattering'' approach, prescribing the asymptotic particle spectrum and the two-particle scattering matrix as an input to the construction. So far, these particles have always been bosons, even if some of the models can be described as deformations of Fermi fields \cite{Allazawi:2012}.

This raises the question whether the approach can encompass asymptotic fermions as well. In fact, in another approach to integrable systems -- the form factor programme \cite{Smirnov:1992}, which tackles the same problem but normally does not address functional analytic questions about the local operators, or alternatively, the convergence of $n$-point functions given by infinite series -- fermionic systems are described by some extension of the usual form factor axioms \cite{BFKZ:sinegordon,Lashkevich:1994}. However, in the operator-algebraic approach, fermionic integrable systems seem so far to have escaped attention.

It is the purpose of the present article to fill that gap, constructing graded-local field algebras by an extension of the techniques used for bosonic models.

To convey the idea, let us recall in rough terms how the case of a scalar boson, of mass $\mu>0$ and with two-particle scattering function $S$, was treated in \cite{Lechner:2008}. Based on the Zamolodchikov-Faddeev operators $z(\theta)$, $\zd(\theta)$ which fulfill an $S$-deformed version of the CCR and act on an $S$-symmetric Fock space, one considers the field
\begin{equation}\label{eq:leftfield}
   \phi(x) = \int d\theta\, \Big(e^{ip(\theta)\cdot x} \zd(\theta) + e^{-ip(\theta)\cdot x} z(\theta)\Big),
\end{equation}
which is analogous to the free Bose field (also in its functional analytic properties), but does not fulfill local commutation relations. Rather, it is interpreted as localized in an extended region, the left wedge $\wedg_L+x$ with tip at $x$. It commutes with the right field, $\phi'(x) = J \phi(-x) J$ where $J$ is a naturally defined antiunitary operator (the PCT operator); thus $\phi'(x)$ is interpreted as localized in the right wedge $\wedg_R+x$ with tip at $x$. One can now pass to the von Neumann algebras generated by smeared versions of these fields, and thus obtain algebras $\M_x$ associated with $\wedg_R+x$, and its commutant $\M_x'$ associated with $\wedg_L+x$. Instead of considering local fields, one may now define local von Neumann algebras
\begin{equation}\label{eq:aintersect}
   \A(\ocal_{x,y}) := \M_x \cap \M_y'  \quad \text{if } \ocal_{x,y} = (\wedg_R+x) \cap (\wedg_L+y). 
\end{equation}
This yields a consistent definition of a quantum field theory. The technically difficult point is now to prove that the intersection \eqref{eq:aintersect} contains more than just multiples of the identity operator; this is usually done via establishing nuclearity of certain inclusion maps. With this settled, one can use standard approaches to scattering theory to establish that the model does indeed have the desired scattering matrix (between bosonic particles).

Our approach to a fermionic version is formally rather simple: We retain the left field $\phi$ as in \eqref{eq:leftfield}, but replace the right field $\phi'$ with a ``twisted'' version,
\begin{equation}
   \hat \phi (x) = i (-1)^N \phi'(x),
\end{equation}
where $N$ is the particle number operator. This field \emph{anti}commutes with $\phi$ if the wedges are spacelike separated. Its smeared versions generate a von Neumann algebra $\M^\twist$ rather than $\M$, and our twisted-local field algebras are simply
\begin{equation}\label{eq:fintersect}
   \F(\ocal_{x,y}) := \M_x^{\twist} \cap \M_y'.
\end{equation}
Most of the remaining construction can be carried over from the bosonic case. An important insight is that nontriviality of the local (field) algebras is guaranteed by the very \emph{same} condition (modular nuclearity) that was used in the bosonic case; since this condition has been established for a certain class of scattering functions $S$ in the literature (at substantial effort), we immediately obtain a large class of nontrivial fermionic models.

While the scalar integrable models of \cite{Lechner:2008} (and their graded modifications) are our prime example, most of our analysis is independent of this specific choice. In fact, most parts can (as in the bosonic case) be considered at a more abstract level: just like \emph{local} 1+1-dimensional quantum field theories can be obtained from a Borchers triple $(\M,U,\Omega)$, consisting of one von Neumann algebra $\M$ (thought to be associated with a wedge), a representation of the translation group $U$, and a vacuum vector $\Omega$, we can describe our models by enhancing these triples with a grading operator $\Gamma$; the integrable models sketched above being examples of this approach, with $\grad=(-1)^N$.

A new aspect of the fermionic situation is how the net of \emph{local} observables, i.e., the even part $\A$ of $\F$ under the grading, is abstractly related to the field net $\F$ itself -- since one may take the view that the theory should be determined by its local observables, not the field net. The usually accepted answer is given by the Doplicher-Haag-Roberts theory of superselection sectors. However, this approach is not applicable to massive 1+1-dimensional theories  \cite{Mueger:massive2d}. Following an analysis of M\"uger \cite{Mueger:qdouble}, we establish however that the net $\A$ lacks Haag duality, in a way that can quite precisely be described; and we find that the field net $\F$, its even part $\A$, and the non-graded algebras in \eqref{eq:aintersect}, can all be seen as fixed points of one larger (though non-local) net of algebras under the action of certain automorphisms. A major part is played by the so-called ``disorder operators'', which also feature in the form factor programme \cite{Lashkevich:1994}.

Our plan for the paper is as follows. After recalling some background in Sec.~\ref{sec:bg}, we introduce graded Borchers triples and the relevant nuclearity condition, and draw some direct consequences, in Sec.~\ref{sec:triples}. The connection with observable algebras is analyzed in Sec.~\ref{sec:obs}. In Sec.~\ref{sec:scattering}, we formulate our variant of Haag-Ruelle scattering theory, which will allow us to identify asymptotic particles as fermions. Sec.~\ref{sec:integrable} then describes the example of scalar integrable models, including the connection of our approach with the form factor programme. We end with a brief outlook in Sec.~\ref{sec:conclusions}.

\section{Prerequisites}\label{sec:bg}

We repeat some technical prerequisites, mainly to fix our notation. For more information, see \cite{Roberts:1970,Rob:lectures2}.

\paragraph{Grading}
A \emph{grading} on a Hilbert space $\Hil$ is given by a unitary involution $\grad$ (i.e., $\grad\in\boundedops$ such that $\Gamma=\Gamma^\ast=\Gamma^{-1}$); we may split $\Hil=\Hil_+ \oplus \Hil_-$ such that $\grad\restriction \Hil_\pm = \pm \idop$. The grading acts on $\boundedops$ by the automorphism $\alpha = \adj \Gamma$. With $\alpha_\pm:=\frac{1}{2}(1 \pm \alpha)$, we can uniquely split $A \in \boundedops$ into ``even'' and ``odd'' parts, $A = A_+ + A_-$ with $\alpha(A_\pm)=\pm A_\pm$; namely, $A_\pm=\alpha_\pm(A)$. We also define the unitary ``twist'' operator $Z:=\frac{1-i}{2} + \frac{1+i}{2} \grad$; it acts as $A^\twist:= ZAZ^\ast= A_+ + i \grad A_-$. 

For a von Neumann algebra $\ncal\subset\boundedops$, set $\ncal^\twist = \{ A^\twist : A \in \ncal\}$. If $\alpha(\ncal)=\ncal$, this ``twist'' reproduces the usual Bose/Fermi  (anti)commutation relations in the following sense: If $A \in\ncal$, $A' \in \ncal^{\twist \prime}=\ncal^{\prime \twist}$, we have 
\begin{equation}
 A'A = AA' - 2 A_-A'_- = \sum_{s,s' = \pm 1} (-1)^{(1-s)(1-s')/4}\alpha_s(A)\alpha_{s'}(A').
\end{equation}
For several operators $A_j\in \ncal_j$ ($j = 1,\ldots,n$), with $\alpha(\ncal_j)=\ncal_j$ and $\ncal_j \subset \ncal_k^{\twist\prime}$ ($j \neq k$), this generalizes to
\begin{equation}\label{eq:gradpermute}
 A_{\sigma(1)} \dotsm A_{\sigma(n)} = \sum_{s_j = \pm 1} \Big(\prod_{\substack{j<k \\ \sigma(j)>\sigma(k)}} (-1)^{(1-s_j)(1-s_k)/4}\Big) \alpha_{s_{1}}(A_{1}) \dotsm \alpha_{s_n}(A_{n})
\end{equation}
for any permutation of $n$ elements, $\sigma\in\perms{n}$.

The even part $\ncal_+ := \{A\in \ncal : \alpha(A)=A\}$ of an algebra $\ncal$ has natural representations $\pi_\pm$ on $\Hil_\pm$, namely, $\pi_\pm(A)=A \restriction \Hil_\pm$.

\paragraph{Minkowski space} 
We work on 1+1-dimensional Minkowski space $\rbb^2$. Define the regions $\wedg_R= \{x : |x^0|<x^1\}$, $\wedg_L=-\wedg_R$. 
For $x\in \rbb^2$, let us define the ``covariant length'' 
\begin{equation}
 \delta(x):=\begin{cases}\sqrt{-x\cdot x} \quad & \text{if }x \in \wedg_R, \\ 0 & \text{otherwise} . \end{cases}
\end{equation}
Then for $x,y \in \rbb^2$ with $\delta(y-x)>0$, we define the double cone $\ocal_{x,y} := (\wedg_R + x) \cap (\wedg_L + y)$; as a specific case, the standard double cone of radius $r$ centered at the origin is $\ocal_r := \ocal_{(0,-r),(0,r)}$. 
For any $\ocal$, we denote by $\ocal'$ its (open) causal complement.

The symmetry group of our spacetime is the proper orthochronous Poincar\'e group $\pcal_+^\uparrow$, consisting of translations $\trans_x$ ($x \in \rbb^2$) and boosts $\boost_\theta$ ($\theta\in\rbb$), and the larger group $\pcal_+$ including also the space-time reflection, $\refl$.  

\paragraph{Nets of algebras}
A net of algebras is an assignment of (open) regions in $\rbb^2$ to von Neumann algebras $\F(\ocal)\subset\boundedops$ such that $\F(\ocal)$ is increasing with $\ocal$ \emph{(isotony)}. We say $\F$ is \emph{local} if $\F(\ocal')\subset \F(\ocal)'$, and \emph{twisted-local} if $\F(\ocal')\subset \F(\ocal)^{\twist\prime}$, for any open $\ocal\subset\rbb^2$.  Given a unitary representation $U$ of $\pcal_+^\uparrow$, we say that $\F$ is \emph{covariant} if $U(g)\F(\ocal)U(g)^\ast = \F(g.\ocal)$ for all $g\in\pcal_+^\uparrow$; similarly for $\pcal_+$, but in this case, we take spacetime reflection to be represented by an \emph{anti}-unitary operator.

We will usually define such nets by specifying only $\F(\ocal_{x,y})$, setting 
\begin{equation}\label{eq:fgen}
  \F(\ocal):=\bigvee_{\ocal_{x,y}\subset\ocal} \F(\ocal_{x,y}) 
\end{equation}
for general open regions $\ocal$ without further mention. In this case, it suffices to check isotony, (twisted) locality or covariance for the $\F(\ocal_{x,y})$ only, in analogous form.

\section{Graded Borchers triples and field algebras}\label{sec:triples}

We start with an abstract operator-algebraic construction, based on the notion of Borchers triples. We introduce a graded variant of these:
\begin{definition}\label{def:triple}
 Given a separable Hilbert space $\Hil$, a \emph{graded Borchers triple} $(\M, U, \Omega, \grad)$ consists of a von Neumann algebra $\M \subset \boundedops$, 
 a unitary representation $U$ of $(\rbb^2,+)$ on $\Hil$, a vector $\Omega \in \Hil$, and a unitary involution $\grad \in \boundedops$, such that
 \begin{enumerate}[(i)]
  \item $U$ is strongly continuous, fulfills the spectrum condition, and has $\Omega$ as a unique (up to a phase) invariant vector;
  \item $\M_x:=U(x) \M U(x)^\ast \subset \mcal$ for all $x \in \wedg_R$;
  \item \label{it:btcyclic} $\Omega$ is cyclic and separating for $\M$;
  \item $\grad \mcal \grad^\ast=\mcal$, $\grad$ commutes with $U$, and $\grad\Omega=\Omega$.
 \end{enumerate} 
\end{definition}
Given such a triple (or rather quadruple), we then set
\begin{equation}\label{eq:fdef}
    \F(\ocal_{x,y}) := \M_x^\twist \cap \M_y^{\prime},
\end{equation}
which generates $\F(\ocal)$ for general $\ocal$ by \eqref{eq:fgen}. 
Analogous to the non-graded case, one immediately has:

\begin{proposition}\label{prop:fnet}
   $\F$ is a twisted-local net of algebras, covariant under an extension of $U$ to a representation of $\pcal_+$ with $U(\boost_\theta)=\Delta^{i\theta/2\pi}$ and $U(\refl)=ZJ$, where $(\Delta,J)$ are the Tomita-Takesaki modular data of $(\M,\Omega)$.
\end{proposition}
\begin{proof}
We first note that, since $\grad \M \grad^\ast=\M$, $\grad\Omega=\Omega$, and due to uniqueness of the modular data, $\grad$ commutes with $\Delta$ and $J$; hence $\Delta Z = Z \Delta$ and $ZJ = JZ^\ast$. 
    
   Now it is well known \cite[Sec.~III]{Borchers:1992} that $U$ extends to a representation of $\pcal_+^\uparrow$ where $U(\boost_{2\pi\theta})=\Delta^{i\theta}$, such that $U(g)\M_xU(g)^\ast=\M_{g.x}$ for all $g$; and $J$ completes this to a representation of $\pcal_+$. Then $U(\refl):=ZJ$ defines another antiunitary involution that combines to a representation of $\pcal_+$ with the same $U(g)$, $g\in\pcal_+^\uparrow$, since $[U(g),Z]=0$. One easily checks $U(\refl)\M^\twist U(\refl) = \M'$, and hence it follows that $U(g) \F(\ocal_{x,y}) U(g)^\ast=\F(g.\ocal_{x,y})$ for \emph{all} $g\in\pcal_+$.

   For twisted locality, one notes that if $\ocal_{z,w}$ is spacelike from $\ocal_{x,y}$, say, $x-w \in \wedg_R$, then
   \begin{equation}
      \F(\ocal_{z,w}) =  \mcal_z^\twist \cap \M_w' \subset \M_w' \subset \M_x' = (\M_x^\twist)^{\twist\prime} \subset (\mcal_x^\twist \cap \M_y')^{\twist\prime} =  \F(\ocal_{x,y})^{\twist\prime};
   \end{equation}
   likewise if $z-y\in\wedg_R$.
\end{proof}

We will introduce a class of examples of such graded triples in Sec.~\ref{sec:integrable}. For the moment, the reader may think of $\mcal'$ being generated by a free Majorana field smeared with test functions in $\wedg_L$ and $\mcal^\twist$ by the same field smeared in $\wedg_R$, with $\Gamma=(-1)^N$.

While in this example, it is clear that the intersection in \eqref{eq:fdef} is nontrivial---$\F(\ocal_{x,y})$ is generated by the Majorana field in $\ocal_{x,y}$---, this does not follow from the abstract axioms; one may well have $\F(\ocal_{x,y}) = \cbb \idop$, even if $\Gamma=\idop$ \cite{LechnerScotford:deformations}. However, an additional condition (``modular nuclearity'') implies in the non-graded case \cite{BuchholzLechner:2004} that $\F(\ocal_{x,y})$ is a type $\mathrm{III}_1$ factor, and has $\Omega$ as a cyclic and separating vector. This is related to the split property, which follows from the modular nuclearity condition \cite{BuchholzDAntoniLongo:1990-1} defined as follows:
\begin{definition}
 We say that a (graded) Borchers triple fulfills \emph{modular nuclearity} (at distance $r \geq 0$) if for any $x \in \rbb^2$ with $\delta(x)>r$, the map
 \begin{equation}
    \Xi_x : \M_x \to \Hil, \quad A \mapsto \Delta^{1/4} A \Omega
 \end{equation}
 is nuclear.
\end{definition}
Based on this condition, which does not depend on the grading, we will establish a graded version of the ``standard split'' property (Proposition~\ref{prop:unitaryequiv} below).\footnote{%
We include a direct proof, inspired by \cite{BDF:universal_structure}, for the convenience of the reader. But Proposition~\ref{prop:unitaryequiv} can also be deduced from the literature: Modular nuclearity is known \cite{BuchholzDAntoniLongo:1990-1} to imply the split property, i.e., existence of a unitary $V$ that fulfills \eqref{eq:vmmv} and \eqref{eq:vabv} \emph{without} grading (replace $\Gamma$ with $\idop$). Now $\adj\Gamma$ is an automorphism of $\M_x$, $\M'_y$, hence one can choose $V$ such that $V \Gamma =(\Gamma\otimes\Gamma)V$ \cite{DopLon:split}. Using this property, $V$ can be used to construct a unitary with the features listed in   Proposition~\ref{prop:unitaryequiv} (``with grading''). 
}

To that end, for the remainder of this section, fix a graded Borchers triple $(\M,U,\Omega,\grad)$ which fulfills modular nuclearity, and $x,y$ with $\delta(y-x)>r$. We first recall the notion of graded tensor product $\otimes_\grad$ of von Neumann algebras \cite[Sec.~2.6]{CKL:superconformal}: For $A \in \M_x'$, $B\in \M_y$, we write $A \otimes_\grad B = A \otimes B_+ + A\Gamma \otimes B_-$, and we denote with $\ncal_{x,y}:=\M_x' \otimes_\grad \M_y \subset \bfk(\Hil \otimes \Hil)$ the von Neumann algebra generated by all such $A \otimes_\grad B$ (linear combinations and closure are sufficient). We note
\begin{lemma}\label{lemma:ncycsep}
   The vector $\Omega\otimes\Omega\in\Hil\otimes\Hil$ is cyclic and separating for $\ncal_{x,y}$.
\end{lemma}
\begin{proof}
  One has $(A \otimes_\grad B) (\Omega\otimes\Omega) = (A \otimes B) (\Omega\otimes\Omega)$, hence the statement follows since $\Omega\otimes\Omega$ is cyclic and separating for the non-graded tensor product $\M_x' \otimes \M_y$. 
\end{proof}

Now we consider the map $\pi_\twist: \ncal_{x,y} \to \boundedops$ given by linear extension of 
\begin{equation}
   \pi_\twist(A \otimes_\grad B) := A^\twist B = ZA Z^\ast B, \quad A \in \M_x', \, B \in \M_y;
\end{equation}
one checks that this is in fact a $\ast$-representation of $\ncal_{x,y}$.
We also consider the state $\omega_\twist := \hscalar{\Omega}{\pi_\twist(\cdotarg)\Omega}$ on $\ncal_{x,y}$.
\begin{lemma}\label{lemma:normal}
    $\omega_\twist$ is normal.
\end{lemma}

\begin{proof}
   Assume without loss of generality that $x=0$ and  $y\in\wedg_R$.
Let $A \in \M_x'=\M'$, $B\in\M_y$. Since $A^\ast\Omega \in \dom \Delta^{-1/4}$, $B\Omega \in \dom \Delta^{+1/4}$, and since $Z$ commutes with $\Delta$, we have
\begin{equation}
   \omega_\twist(A \otimes_\grad B) = \hscalar{Z \Delta^{-1/4} A^\ast \Omega}{ \Delta^{1/4 }B \Omega} = \hscalar{Z\Psi(A^\ast)}{\Xi_y(B)},
\end{equation}
where the map $\Psi: \M' \to \Hil$, $A \mapsto \Delta^{-1/4} A \Omega$ is normal and bounded (cf.~\cite[p.~236]{BuchholzDAntoniLongo:1990-1}) and the map $\Xi_y: \mcal_y \to \Hil$ is nuclear by assumption. Writing $\Xi_y=\sum_j \ell_j(\cdotarg) \psi_j$ with $\psi_j \in \Hil$ and normal functionals $\ell_j$,
we have $\omega_\twist(A \otimes_\grad B) = \sum_j \hscalar{Z\Psi(A^\ast )}{\psi_j} \ell_j(B)$ as a norm-convergent sum. 
Each summand extends to a normal functional on $\ncal_{x,y}$.
Thus $\omega_\twist$ is a norm-convergent sum of normal functionals and hence normal.
\end{proof}

\begin{lemma}\label{lemma:faithful}
    $\omega_\twist$ is faithful.
\end{lemma}

\begin{proof}
  Since $\omega_\twist$ is normal (Lemma~\ref{lemma:normal}) and $\ncal_{x,y}$ has a separating vector (Lemma~\ref{lemma:ncycsep}), it follows that $\omega_\twist= \hscalar{\hat\Omega}{ \cdotarg \hat\Omega}\restriction \ncal_{x,y}$ with some $\hat\Omega\in\Hil$ \cite[Theorem 7.2.3]{KadRin:algebras2}. Replacing $\hat\Omega$ with $V\hat\Omega$ with a suitable unitary $V\in\ncal_{x,y}'$, we can assume that $\hscalar{\Omega\otimes\Omega}{\hat\Omega} \neq 0$ (in view of Lemma~\ref{lemma:ncycsep}).

  We show that $\hat\Omega$ is separating for $\ncal_{x,y}$ (and hence $\omega_\twist$ faithful): Since also $\delta(y-z-x)>r$ for $z$ from some small open set in $\wedg_R$, we can assume that $\omega_\twist$ extends to $\ncal_{x,y-z}$ and is still given there by $\hscalar{\hat\Omega}{\cdotarg \hat\Omega}$. Now let $C \in \ncal_{x,y}$ such that $C\hat\Omega=0$. By translation invariance of $\Omega$, we have $C (U \otimes U)(z) \hat \Omega=0$ for $z$ in some open set. By an analyticity argument (using the spectrum condition for $U$), we find $C (U \otimes U)(z) \hat \Omega=0$ for \emph{all} $z\in\rbb^2$. Taking $z$ to spacelike infinity, $(U\otimes U)(z)$ converges weakly to the projector $P$ onto $\Omega\otimes\Omega$ (due to uniqueness of the invariant vector); hence $CP\hat\Omega=0$. Now $P\hat \Omega \neq 0$ is separating for $\ncal_{x,y}$ (Lemma~\ref{lemma:ncycsep}), hence $C=0$.
\end{proof}

Now we can establish a graded analogue to the ``standard split'' property \cite{DopLon:split}:
\begin{proposition}\label{prop:unitaryequiv}
    There exists a unitary $V:\Hil \to \Hil\otimes\Hil$ such that
    \begin{align}
    V \big( \M_x^{\twist\prime} \vee \M_y \big) V^\ast &= \M_x' \otimes_\grad \M_y, \label{eq:vmmv} \\  
       V A^\twist B V^\ast &= A \otimes_\grad B \quad \text{for all } A \in \M_x',\, B \in \M_y, \label{eq:vabv} \\   
       V \grad &= (\grad \otimes \grad) V.  \label{eq:vg}  
    \end{align}
    Further, $\Omega$ is cyclic and separating for $\M_x^{\twist\prime} \vee \M_y$.
\end{proposition}
\begin{proof}
  Since $\omega_\twist$ is faithful and normal, $\pi_\twist$ is a $\ast$-isomorphism. Also, both $\ncal_{x,y}$ and $\pi_\twist(\ncal_{x,y})=\M_x^{\twist\prime} \vee \M_y$ have a cyclic and separating vector, namely $\Omega\otimes \Omega$ and $\Omega$ respectively (the latter is separating since $\omega_\twist$ is faithful). Hence $V$ fulfilling \eqref{eq:vmmv}, \eqref{eq:vabv} exists by the unitary implementation theorem \cite[Theorem 7.2.9]{KadRin:algebras2}. In fact, if we choose $V$ to be the unitary standard implementation \cite[Appendix]{DopLon:split}, then also \eqref{eq:vg} is fulfilled (due to uniqueness of the standard implementation, and since $\grad\Omega=\Omega$ and $(\grad\otimes\grad)(\Omega\otimes\Omega)=\Omega\otimes\Omega$).   
\end{proof}

We can now conclude, mirroring \cite[Theorem~2.5]{Lechner:2008} (see also \cite[Theorem 3.7]{BostelmannCadamuro:examples} for proof techniques):

\begin{theorem}\label{thm:gradednet}
Let $(\M,U,\Omega,\grad)$ be a graded Borchers triple which fulfills modular nuclearity at distance $r$. Then we have
\begin{enumerate}[(a)]
 \item \label{it:locwedge} \emph{Locally generated wedge algebras:} $\F(\wedg_R)=\M^\twist$, $\F(\wedg_L)=\M^{\prime}$;
\end{enumerate}
and for any $x,y$ with $\delta(y-x)>r$:
\begin{enumerate}[(a)]
\setcounter{enumi}{1}
 \item \label{it:reehschlieder} \emph{Reeh-Schlieder property:} $\Omega$ is cyclic and separating for $\F(\ocal_{x,y})$;
 \item \label{it:haagdual} \emph{Twisted Haag duality:} $\F(\ocal_{x,y})'=\F(\ocal_{x,y}')^\twist$;
 \item \label{it:weakadd} \emph{Weak additivity:} If $e\in\rbb^2$ is spacelike or lightlike, then $\bigvee_{t \in \rbb} \F(\ocal_{x,y}+te) = \boundedops$.
\end{enumerate}
 
\end{theorem}

\begin{proof}

We first show (\ref{it:reehschlieder}): By Proposition~\ref{prop:unitaryequiv}, $\Omega$ is cyclic and separating for $\M_x^{\twist\prime} \vee \M_y$
and hence also for $(\M_x^{\twist\prime} \vee \M_y)'=\M_x^\twist \cap \M_y'=\F(\ocal_{x,y})$.

In (\ref{it:locwedge}), we prove the equality for $\M^\twist$, the other follows from covariance. First note that the modular operator for $(\M^\twist,\Omega)$ is $Z \Delta Z^\ast=\Delta$. Then consider the subspace $\hcal_R :=  \F(\wedg_R) \Omega \subset \M^\twist\Omega\subset \operatorname{dom} \Delta^{1/2}$, which is dense in $\Hil$ by (\ref{it:reehschlieder}), and invariant under $\Delta^{it}$ by Lorentz covariance of $\F$. Then $\hcal_R$ is a core for $\Delta^{1/2}$, cf.~\cite[Ch.~II Prop.~1.7]{EngelNagel:semigroups}, which implies that $\F(\wedg_R) \subset \M^\twist$ is strong-operator dense \cite[Theorem~9.2.36]{KadRin:algebras2}.

Regarding (\ref{it:haagdual}), observe that 
\begin{equation}
\begin{split}
 \F(\ocal_{x,y})' = \M_x^{\twist\prime} \vee \M_y = \F((\wedg_L+x)  \cup (\wedg_R+y))^\twist = \F(\ocal_{x,y}')^\twist,
\end{split}
\end{equation}
where the second equality makes use of (\ref{it:locwedge}).

Part (\ref{it:weakadd}) follows by a standard argument \cite[Appendix]{Kuckert:regions}: Let $A \in \F(\ocal_{x,y})$ and $B \in (\bigvee_{t\in\rbb} \F(\ocal_{x,y} + te))'$. The spectrum condition for $U$ implies that $f(t):=\hscalar{A\Omega}{U(-te) B \Omega} = \hscalar{B^\ast\Omega}{U(te) A^\ast \Omega}$ has bounded analytic continuations to both the upper and the lower halfplanes; hence $f$ must be constant. Now as $t \to \infty$, due to uniqueness of $\Omega$, the unitaries $U(te)$ weakly converge to the projector onto $\Omega$; therefore $\hscalar{A\Omega}{B \Omega} = \hscalar{A \Omega}{\Omega} \hscalar{\Omega}{B\Omega}$. Then (\ref{it:reehschlieder}) implies $B\Omega = \hscalar{\Omega}{B\Omega} \Omega$ and further $B = \hscalar{\Omega}{B\Omega} \idop$, showing the claim.
\end{proof}

\section{Observable algebras and disorder operators}\label{sec:obs}

We now focus our attention on the ``observable algebras'' defined by 
\begin{equation}
\A(\ocal_{x,y}):=\F(\ocal_{x,y})_+. 
\end{equation}
This automatically yields a $\pcal_+$-covariant, local (rather than twisted-local) net, as one easily checks. Assuming modular nuclearity, it also fulfills the Reeh-Schlieder property in its two naturally given representations $\pi_+$ and $\pi_-$ with respect to $\Omega$ and $F\Omega$, respectively, where $F$ is any fixed invertible element of $\F(\ocal_{x,y})_-$.

However, the net $\A$ typically lacks Haag duality; that is, in general, $\pi_\pm(\A(\ocal)) \subsetneq \pi_\pm(\A(\ocal'))'$. This may roughly be seen as follows: pick two odd operators $F_L$ and $F_R$ in $\F(\wedg_L+x)$ and $\F(\wedg_R+y)$ respectively; then the product $F_LF_R$ is even and commutes with $\A(\ocal_{x,y})$, but is not contained in $\A(\ocal_{x,y}')=\F(\wedg_{L}+x)_+\vee \F(\wedg_{R}+y)_+$, since the individual factors $F_L$, $F_R$ are not even. Alternatively, if one can find an even unitary $V$ which \emph{anti}commutes with $\F(\wedg_L+x)_-$ but \emph{commutes} with $\F(\wedg_R+y)$ (that is, $\operatorname{ad} V = \operatorname{ad} \grad$ ``on the left'' but $\operatorname{ad} V = \operatorname{id}$ ``on the right''), then this operator is contained in $\A(\ocal_{x,y}')'$ but not in $\A(\ocal_{x,y})$. Such $V$ we will call \emph{disorder operators}.\footnote{The nomenclature stems from the analogy with the disorder operators in the Ising model obtained as a lattice limit \cite{SchroerTruong:1978}, although we stress that this model arises via \emph{non-twisted} Borchers triples \cite{Lechner:2008}.}

We will clarify this phenomenon in the following; it is essentially a special case of the analysis in \cite{Mueger:qdouble}, although the arguments are more direct in our context. We define formally:
\begin{definition}\label{def:disorder}
Given a double cone $\ocal_{x,y}$, we call $V\in\boundedops$ a \emph{left [right] disorder operator} in $\ocal_{x,y}$ if $V$ is unitary, commutes with $\grad$, and if 
\begin{equation}\label{eq:leftdis}
\begin{aligned}
   V A V^\ast &= \grad A \grad^\ast  & \text{for all $A \in \M_x'$ [all $A \in \M_y$]}, 
   \\
   V A V^\ast &= A &\text{for all $A \in \M_y$ [all $A \in \M_x'$]}.
\end{aligned}
\end{equation}
We denote the set of these $V$ as $\diso_L(\ocal_{x,y})$ [$\diso_R(\ocal_{x,y})$]. 
\end{definition}
Since $V$ commutes with $\grad$, we may replace $\M_x'$ and/or $\M_y$ in the conditions \eqref{eq:leftdis} above with their twisted counterparts. The $\diso_{L/R}(\ocal_{x,y})$ are selfadjoint sets, but not usually algebras.  We further note:

\begin{lemma}\label{lemma:disorder}
The sets of disorder operators fulfill:
 \begin{enumerate}[(a)]
  \item \label{it:visot} If $\ocal_{x,y} \subset \ocal_{z,w}$ then $\diso_{L}(\ocal_{x,y}) \subset \diso_L(\ocal_{z,w})$;
  \item \label{it:vpoi} For $g \in \pcal_{+}^\uparrow$, we have $U(g)\diso_{L}(\ocal_{x,y})U(g)^\ast = \diso_{L}(\ocal_{g.x,g.y})$; 
  \item \label{it:vrefl} $\diso_R(\ocal_{x,y}) = U(\refl) \diso_L(\ocal_{-y,-x}) U(\refl)$;
  \item \label{it:gammalr} $\diso_R(\ocal_{x,y}) = \grad \diso_L(\ocal_{x,y})$;
  \item \label{it:vlocaltf} If $V \in \diso_{L}(\ocal_{x,y})$, then $V \F(\ocal_{x,y}) V^\ast \subset \F(\ocal_{x,y})$;
  \item \label{it:vlunique} If $V,\hat V \in \diso_{L}(\ocal_{x,y})$, then $V = A \hat V$ with some unitary $A \in \A(\ocal_{x,y})$;
\item \label{it:vnontriv} If the Borchers triple fulfills the modular nuclearity condition with distance $r$, and $\delta(y-x)>r$, then $\diso_{L}(\ocal_{x,y})$ is not empty.

  \end{enumerate}
  \smallskip
 
 \noindent
 Corresponding statements hold with the roles of $\diso_L$ and $\diso_R$ reversed.
\end{lemma}

\begin{proof} 
 (\ref{it:visot}), (\ref{it:vpoi}), (\ref{it:vrefl}) follow from covariance properties of $\M_x$ as in Proposition~\ref{prop:fnet}, and (\ref{it:gammalr}), (\ref{it:vlocaltf}) are clear from Definition~\ref{def:disorder}. For (\ref{it:vlunique}), one notes that $V \hat V^\ast$ commutes with $\M_x^{\twist\prime} \vee \M_y = \F(\ocal_{x,y})^{\prime}$, hence $V \hat V^\ast\in\F(\ocal_{x,y})$; but $V \hat V^\ast$ is even by definition.   
 For (\ref{it:vnontriv}), set $V_L:= V^\ast (\grad \otimes \idop)V$ with $V$ the unitary from Proposition~\ref{prop:unitaryequiv}; one checks that $V_L \in\diso_{L}(\ocal_{x,y})$.
\end{proof}

In particular, there is at most one (and under nuclearity assumptions, exactly one) left/right disorder operator up to ``local multiples''.
We now discuss the relation between $\A(\ocal_{x,y}')$ and $\F(\ocal_{x,y}')$. 

\begin{proposition}\label{prop:adual}
Suppose that $\diso_L(\ocal_{x,y}) \neq \emptyset$. Then
\begin{enumerate}[(a)]
\item \label{it:afv}  $\A(\ocal_{x,y}')' = \F(\ocal_{x,y}) \vee \diso_L(\ocal_{x,y}) \vee \diso_R(\ocal_{x,y})$,
\item \label{it:anondual}  $\pi_\pm(\A(\ocal_{x,y}'))' = \pi_\pm \big( \A(\ocal_{x,y}) \vee \diso_L(\ocal_{x,y}) \big)$.
\end{enumerate}
\end{proposition}
Thus the (essentially unique) disorder operator is the reason for Haag duality to fail (unless, of course, $\grad=\idop$, in which case $\idop$ is a disorder operator, and Haag duality holds).

\begin{proof}
It is clear that
\begin{equation}\label{eq:afxy}
   \A(\ocal_{x,y}') = \A(\wedg_{L}+x) \vee \A(\wedg_R+y) ,
\quad
   \F(\ocal_{x,y}') = \F(\wedg_{L}+x) \vee \F(\wedg_R+y) .
\end{equation}
Now pick $V_L\in\diso_L(\ocal_{x,y})$ and set $V_R:=\grad V_L \in\diso_R(\ocal_{x,y})$ (Lemma~\ref{lemma:disorder}(\ref{it:gammalr})). Defining 
\begin{equation}
   m(A) := \frac{1}{4} \big( A + V_L A V_L^\ast + V_R A V_R^\ast + V_L V_R A V_R^\ast V_L^\ast \big),
\end{equation}
one verifies using \eqref{eq:afxy} that $m$ is a conditional expectation from $\F(\ocal_{x,y}')$ onto $\A(\ocal_{x,y}')$.
Now for $A \in \F(\ocal_{x,y}')$, it holds that
\begin{equation}
   A \in \A(\ocal_{x,y}') \quad\Leftrightarrow\quad m(A)=A \quad\Leftrightarrow\quad [A,V_L] = 0=[A,V_R].
\end{equation}
Thus
\begin{equation}
   \A(\ocal_{x,y}') = \F(\ocal_{x,y}') \cap \{V_L\}' \cap \{V_R\}',
\end{equation}
and taking commutants yields (\ref{it:afv}) together with Lemma~\ref{lemma:disorder}(\ref{it:vlunique}).---Now (\ref{it:anondual}) follows using Lemma~\ref{lemma:disorder}(\ref{it:gammalr}), since $\pi_\pm(\grad)=\pm 1$. 
\end{proof}

To put this in a somewhat more natural context, we follow \cite{Mueger:qdouble} to define the extended algebras
\begin{equation}
   \hat \F(\ocal_{x,y}) := \F(\ocal_{x,y}) \vee \diso_L(\ocal_{x,y}),
\quad
\hat \A(\ocal_{x,y}) := \hat\F(\ocal_{x,y})_+.
\end{equation}
One checks using Lemma~\ref{lemma:disorder}(\ref{it:visot}),(\ref{it:vpoi}) that $\hat\F$ extends to an isotonous, $\pcal_+^\uparrow$-covariant net, although it is in general neither local nor twisted local. Likewise, $\hat\A$ extends to an isotonous, $\pcal_+$-covariant, and \emph{local} net (using Lemma~\ref{lemma:disorder}(\ref{it:vrefl}) and the fact that the disorder operators commute with all even operators).
Further, with modular nuclearity assumed, Lemma~\ref{lemma:disorder}(\ref{it:vnontriv}) and Proposition~\ref{prop:adual}(\ref{it:anondual}) imply that for sufficiently large double cones,
\begin{equation}
  \pi_\pm(\hat\A(\ocal_{x,y})) = \pi_\pm(\A(\ocal_{x,y}'))';
\end{equation}
that is, $\hat \A$ is the dual net of $\A$ there; in particular, $\hat \A$ fulfills Haag duality for these double cones in the sectors $\pi_\pm$.

Apart from the automorphism $\alpha = \adj \grad$ on $\hat\F$, we can define another automorphism $\beta$ as follows. 
Fixing a particular $V \in \diso_L(\ocal_{x,y})$, every $A \in \hat\F(\ocal_{x,y})$ has a unique decomposition of the form
\begin{equation}
   A = A_1 + A_2 V, \quad A_{1,2} \in \F(\ocal_{x,y}).
\end{equation}
We then define
\begin{equation}
   \beta(A) := A_1 - A_2 V.
\end{equation}
This map $\beta$ is in fact independent of the choice of $V$ (Lemma~\ref{lemma:disorder}(\ref{it:vlunique})) 
and of the choice of double cone (Lemma~\ref{lemma:disorder}(\ref{it:visot})), 
hence we obtain a linear map on the quasilocal algebra of $\hat \F$; with Lemma~\ref{lemma:disorder}(\ref{it:vlocaltf}),(\ref{it:vlunique}), one sees that $\beta$ is a $\ast$-automorphism. 

We thus find a simple example of a ``quantum double'' symmetry in the sense of Drinfel'd \cite{Drinfeld:double,Mueger:qdouble}; however, while $\alpha$ is unitarily implemented with invariant vacuum vector, $\beta$ is not, i.e., this part of the symmetry group is ``spontaneously broken'' \cite{Mueger:qdouble}. 

It is amusing that the \emph{untwisted} local net $\check \A(\ocal_{x,y}) := \mcal_x \cap \mcal_y'$ appears as a fixed point of the same (nonlocal) net as our twisted $\A(\ocal)$: 
\begin{proposition}\label{prop:fixpoints}
 Suppose that the graded Borchers triple fulfils modular nuclearity. Then for any double cone $\ocal=\ocal_{x,y}$ with $\delta(y-x)>r$, 
 one has
 \begin{align}\label{eq:ahatfix}
    \hat \A(\ocal)&= \{ A \in \hat\F(\ocal) : \alpha (A) = A\}, 
    \\ \label{eq:ffix}
    \F(\ocal) &= \{ A \in \hat\F(\ocal) : \beta (A) = A \},
\\ \label{eq:acheckfix}
    \check\A(\ocal) &= \{ A\in \hat\F(\ocal) : (\alpha\circ\beta) (A) = A\} ,
\\ \label{eq:afix}
    \A(\ocal) &= \{ A \in \hat\F(\ocal) : \alpha(A)=\beta (A) = A\} .
 \end{align} 
\end{proposition}

\begin{proof}
  The statements for $\hat \A$, $\F$, $\A$ follow directly from the definitions (even without modular nuclearity). For $\check\A$, pick $V \in \diso_L(\ocal_{x,y})$ (by Lemma~\ref{lemma:disorder}(\ref{it:vnontriv})); then it is clear that
  \begin{equation}
       \{A \in \hat\F(\ocal) : (\alpha\circ\beta) (A) = A\} = \F(\ocal_{x,y})_{+} + \F(\ocal_{x,y})_{-} V.   
  \end{equation}
  The even parts of this expression and of $\check A(\ocal)$ clearly agree. Now let $A \in \F(\ocal_{x,y})_{-}$; then $AV\in \mcal_{y}'$ and both $A$ and $V$ (anti)commute with even (odd) elements of $\mcal_x'$, thus $AV\in \mcal_x \cap \mcal_y' = \check \A(\ocal_{x,y})$. On the other hand, let $A\in \check \A(\ocal_{x,y})$ be odd; then $B:=AV^\ast \in \mcal_y'$ and both $A$ and $V$ (anti)commute with even (odd) elements of $\mcal_x^{\twist\prime}$, thus $B \in  \F(\ocal_{x,y})_{-}$.
\end{proof}

After these clarifications on the structure of the observable net $\A$, it would now be of interest whether the field net $\F$, of which $\A$ is the even part, can be uniquely recovered from $\A$ by a canonical, model-independent construction. In higher space-time dimensions, the approach to this question would be a Doplicher-Haag-Roberts (DHR) type charge analysis (\cite{DHR:particle_statistics_1}; see \cite{Haa:LQP,Rob:lectures2} for reviews).  

In our current situation, however, this approach leads to technical difficulties: The DHR analysis (as well as its variants) rely on Haag duality of $\A$ as a central ingredient. As we have seen above, this condition is generically violated in our context, as also noted in \cite{Mueger:qdouble}. One might resort to considering the dual net $\A^d=\hat \A$ instead; but apart from the question whether $\A$ and $\hat \A$ describe ``the same physics'', another technical problem arises: in 1+1 space-time dimensions, any Haag-dual net which also has the split property for wedges (as $\hat \A$ would \cite[Sec.~4.3]{Mueger:qdouble}) cannot have nontrivial localized DHR sectors \cite{Mueger:massive2d}. In the following, we will therefore take the field net $\fcal$ defined by \eqref{eq:fdef} as the ``correct'' one for our purposes.

\section{Scattering theory}\label{sec:scattering}

We now deal with the question in which sense our nets describe scattering particles, and in particular, whether these are fermions or bosons. The usual approach would be to apply Haag--Ruelle scattering theory (see, e.g., \cite{Araki:qft}) to the ``observable'' net $\A$ acting on the Hilbert space $\Hil_+$. However, in general (and certainly in the examples discussed in Sec.~\ref{sec:integrable}, or indeed in the free Majorana field), the representation of $\pcal_+^\uparrow$ on $\Hil_+$ does not have an isolated mass shell in its energy-momentum spectrum, i.e., it lacks identifiable single-particle states. In $3+1$-dimensional spacetime, one would therefore resort to a generalization of scattering theory on the level of DHR charge sectors and the field algebras constructed from them, as outlined in \cite[Sec.~VII]{DHR:particle_statistics_2}. Since DHR analysis is however not applicable in our context, as discussed in Sec.~\ref{sec:obs}, we will directly work with the field net $\F$ in \eqref{eq:fdef} instead.

To that end, we will follow the recent approach of Duell \cite{Duell:wedge_scattering} to Haag--Ruelle scattering based on wedge-localized observables, specializing it to 1+1 space-time dimensions. We formulate it here directly for the case of Borchers triples, which comes with no loss of generality, since all 1+1-dimensional local nets which fulfill wedge duality (as assumed in \cite{Duell:wedge_scattering}) arise from such triples \cite{Borchers:1992}. Also, we generalize the framework to the graded case, which is possible with minimal modifications. 

For scattering theory, we need the following extra assumption on the translation group unitaries $U$, which (due to strong continuity) we can write as $U(x)=\exp( i P \cdot x)$; the joint spectral measure of the 2-vector-valued generator $P$ will be denoted as $E$.
\begin{definition}
We say that the graded Borchers triple $(\M,U,\Omega,\grad)$ has a \emph{mass gap} (at mass $\mu>0$) if, for some $M>\mu$, 
\begin{equation}
   \supp E \subset \{0\} \cup \{p: p\cdot p = \mu^2 \} \cup \{p: p\cdot p > M^2 \} .
\end{equation}
\end{definition}

This will be our standing assumption in the following. With $E_\mu := E(\{p: p\cdot p = \mu^2 \} )$, we define the ``single-particle space'' $\kcal:=E_\mu \Hil$. As all $U(g)$, $g \in \pcal_+$, commute with $E_\mu$, we have a representation $g\mapsto U(g)\restriction \kcal$ on that space (which we again denote as $U$). We also consider the (unsymmetrized) Fock space $\fock(\kcal)$ with a representation $U_\fock$ of $\pcal_+$, where $U_\fock(g)$ is the usual (unsymmetrized) second quantization of $U(g)\restriction\kcal$ for $g \in \pcal_+^\uparrow$, and 
 \begin{equation}
    U_\fock(\refl) (\psi_1 \otimes \ldots \otimes \psi_n) := U(\refl) \psi_n \otimes \dotsm \otimes U(\refl) \psi_1.
\end{equation}

We also consider rapidity-ordered subspaces as follows. For $\psi \in \kcal$, define the rapidity support of $\psi$ as 
\begin{equation}
   \raps{\psi}:=\operatorname{cch} \{\operatorname{artanh}(p^0/p^1): p \in \supp E(\cdot)\psi \},
\end{equation}
where $\operatorname{cch}$ denotes the closed convex hull. Correspondingly, for $f \in \scal(\rbb^2)$, set 
\begin{equation}
   \raps{f}:= \operatorname{cch} \{\operatorname{artanh}(p^0/p^1) : p \in \supp \tilde f\}.
\end{equation}
Thus rapidity supports are closed intervals. The usual partial order $\prec$ of intervals then lifts to a partial order $\prec$ (``precursor relation'') on $\kcal$ and on $\scal(\rbb^2)$. We now define the ``rapidity ordered'' Fock space $\fock^\prec(\kcal)$ as the closed span of vectors $\psi_1 \otimes \dotsm \otimes\psi_n$ with $\psi_1 \prec \dotsm \prec \psi_n$; similarly for $\fock^{\succ}(\kcal)$. Note that $U_\fock(\refl)\fock^{\succ}(\kcal)=\fock^{\prec}(\kcal)$. 

In the following, we will often choose a Schwartz function $\chi$ such that the support of its Fourier transform $\tilde \chi$ is compact and intersects the support of $E$ only on the mass shell, $p\cdot p = \mu^2$. (For short, we shall say that ``$\tilde \chi$ has support near the mass shell''.) For $A \in \boundedops$, we set $A^\chi=\int d^2x\, \chi(x) U(x) AU(x)^\ast$. Further for $\raps{f}$ compact, we define
\begin{equation}\label{eq:achidef}
   A^\chi_\tau(f) := \int dx\, f(\tau,x) \, U(\tau,x) A^\chi U(\tau,x)^\ast \quad 
   \text{where} \; f(\tau,x):=\int \frac{dk}{2\pi}  \tilde f(k) \exp \big( ixk-i\tau \sqrt{k^2+\mu^2}\big).
\end{equation}

We can now summarize, and slightly extend, the results of \cite{Duell:wedge_scattering} as follows:

\begin{theorem} \label{thm:wedgescatter}
There exist isometries $\mo_\scout: \fock^\prec(\kcal) \to \Hil$ and $\mo_\scin: \fock^\succ(\kcal)\to\Hil$ with the following property. 
Let $x,y\in\rbb^2$ and $A_1,\ldots,A_n \in \M_y'$, $\hat A_1,\ldots,\hat A_n \in \M_x^\twist$ such that $E_\mu A_j \Omega = E_\mu \hat A_j \Omega$. 
Let $f_1,\ldots,f_n\in\scal(\rbb^2)$ such that $\raps{f_1} \prec \dots \prec \raps{f_n}$, and let $\tilde\chi$ have support near the mass shell.
Then 
\begin{equation}\label{eq:moconv}
 \begin{aligned}
     \lim_{\tau \to \infty} A_{1,\tau}^\chi(f_1) A^\chi_{2,\tau}(f_2)\dotsm A^\chi_{n,\tau}(f_n)\Omega &= \mo_\scout \big(A^\chi_{1}(f_1)\Omega \otimes \dotsm \otimes A^\chi_{n}(f_n)\Omega \big),
\\
     \lim_{\tau \to -\infty} A^\chi_{n,\tau}(f_n) \dotsm  A^\chi_{2,\tau}(f_2) A^\chi_{1,\tau}(f_1)\Omega &= \mo_\scin \big(A^\chi_{n}(f_n)\Omega \otimes \dotsm \otimes A^\chi_{1}(f_1)\Omega \big).
\end{aligned} 
\end{equation}
Further we have
\begin{equation}\label{eq:mosymm}
   U(g) \mo_\scinout = \mo_\scinout U_\fock(g) \; \text{ for all $g \in \pcal_+^\uparrow$}, 
   \qquad
   U(\refl) \mo_\scout = \mo_\scin U_\fock(\refl). 
\end{equation}
\end{theorem}

The proof is as in \cite{Duell:wedge_scattering}, in particular Theorem 6, Proposition 23, and Theorem 24 there, with the following changes: Apart from the easy addition of the relation for $U(\refl)$ in \eqref{eq:mosymm}, and trivial changes in notation, we have replaced $\mcal_x$ with $\mcal_x^\twist$.
This can be accommodated as follow: Corollary~10 in  \cite{Duell:wedge_scattering} holds analogously for the (anti-)commutator of the even and the odd parts of the operators involved; then, the estimates for approximations of scattering states can be done analogously, utilizing Eq.~\eqref{eq:gradpermute} instead of spacelike commutation relations. Also the crucial Lemma 3 (existence of swappable vectors) carries over. We explain these techniques in some more detail in Lemma~\ref{lemma:swapapprox} and Proposition~\ref{prop:improvedconv} below.

The above result does not depend on the modular nuclearity condition. With that given, however, we can pass to the usual scattering theory with asymptotic Bose/Fermi states. To that end, consider the following representation of the permutation group $\perms{n}$ on $\kcal^{\otimes n}$,
\begin{equation}
   \pi_\grad(\sigma) ( \psi_1 \otimes \dotsm \otimes \psi_n ) =  \prod_{\substack{i < j \\ \sigma(i)>\sigma(j)}} (-1)^{ (1-\Gamma_{\sigma(i)})(1-\Gamma_{\sigma(j)})/4 } \psi_{\sigma(1)} \otimes \dotsm \otimes \psi_{\sigma(n)}, 
\end{equation}
where $\grad_j$ multiplies with $\grad$ in the $j$-th factor; let $\fock^{\grad}(\kcal)\subset\fock(\kcal)$ be the corresponding invariant subspace, the projector onto which is $P_\grad=\bigoplus_{n\geq 0}\frac{1}{n!} \sum_{\sigma\in\perms{n}} \pi_\grad(\sigma)$. Then we obtain:

\begin{theorem}\label{thm:localscatter}
  Suppose that the graded Borchers triple $(\M,U,\Omega,\grad)$ has a mass gap and fulfills modular nuclearity. Then the operators $\mo_\scinout$ uniquely extend to densely defined linear operators $\fock(\kcal) \to \Hil$ such that, in the notation of Theorem~\ref{thm:wedgescatter}, if $A_j=\hat A_j$ and the $f_j$ have \emph{disjoint} (not necessarily ordered) rapidity support, 
  then convergence as in \eqref{eq:moconv} holds. We have $\mo_\scinout P_\grad = \mo_\scinout$. Further, $\mo^\Gamma_{\scinout} := \mo_{\scinout}(N!)^{-1/2} \restriction \fock^\grad(\kcal) $ is isometric.
\end{theorem}
In this case, we can also define the \emph{scattering operator} $S:=(\mo^\grad_\scout)^\ast \mo_\scin^\grad:\fock^\grad(\kcal) \to \fock^\grad(\kcal)$.

\begin{proof}
 We treat only $\mo_\scout$, the case of $\mo_\scin$ being analogous.
 Let $\psi_1,\ldots,\psi_n\in\kcal$ such that $\psi_1,\ldots,\psi_n$ have disjoint rapidity support. Then there is a unique permutation $\sigma$ such that $\psi_{\sigma(1)}\prec \dotsm \prec\psi_{\sigma(n)}$; hence setting $\psi:=\psi_1 \otimes \dotsm \otimes \psi_n$, we have  $\pi_\grad(\sigma)\psi\in \fock^\prec(\kcal)$. We put
 \begin{equation}
      \mo_\scout \psi := \mo_\scout \pi_\grad(\sigma)\psi,
 \end{equation}
 thus defining $\mo_\scout$ on a total set, and with $\mo_\scout \pi_\grad(\sigma)=\mo_\scout$ for every permutation $\sigma\in\perms{n}$ (on $n$-particle vectors). For $\psi\in\fock^\prec(\kcal)$, we then have $\mo_\scout \psi = \mo_\scout P_\grad \psi$. On the other hand, using that $\pi_\grad(\sigma)\psi \perp \psi$ if $\sigma\neq \id$ (due to disjoint spectral supports), one computes that 
 \begin{equation}
    \|  P_\grad \psi \|^2 = \sum_{\sigma,\sigma'} \frac{1}{(n!)^2} \hscalar{ \pi_\grad(\sigma)\psi}{ \pi_\grad(\sigma')\psi}
    = \frac{1}{n!} \|\psi\|^2 = \frac{1}{n!} \|\mo_\scout \psi\|^2;
 \end{equation}
 using isometry of $\mo_\scout$ on $\fock^\prec(\kcal)$.
 Thus $\| \mo_\scout P_\grad \psi \|= \sqrt{n!} \| P_\grad \psi \|$. Noting that $P_\grad\fock^\prec(\kcal)$ is dense in $\fock^\grad(\kcal)$, we get isometry of $\mo_\scout^\grad$. 
 
 Further, let $A_1,\ldots,A_n \in \F(\ocal_{x,y})$ for some $x,y$, and $f_1\prec \dots \prec f_n$.  For simplicity, we can assume that each $A_j$ is either even or odd. Then, using \eqref{eq:gradpermute} and a generalization of \cite[Corollary~10]{Duell:wedge_scattering} to the relevant (anti)commutators, with $A_j=\hat A_j$, one shows that for each permutation $\sigma$ and with a certain $s_\sigma\in\{\pm 1\}$,
 \begin{equation}
     \lim_{\tau\to\infty} A_{\sigma(1),\tau}^\chi(f_{\sigma(1)}) \dotsm A^\chi_{\sigma(n),\tau}(f_{\sigma(n)}) \Omega
  =  s_\sigma \lim_{\tau\to\infty} A_{1,\tau}^\chi(f_1) \dotsm A^\chi_{n,\tau}(f_n) \Omega .
\end{equation}
 But the right hand side is just $\mo_\scout \pi_\grad(\sigma^{-1})   A_{\sigma(1)}^\chi(f_{\sigma(1)})\Omega \otimes \dotsm \otimes A_{\sigma(n)}^\chi(f_{\sigma(n)})\Omega$ by our definitions.
 Finally, and most importantly, modular nuclearity guarantees that this relation fixes $\mo_\scout$ uniquely, since sufficiently many $A_j=\hat A_j \in \F(\ocal_{x,y})$ can be chosen.
\end{proof}

The above construction is possible for every Borchers triple; however, it may be difficult to work out the action of the operators $A^\chi_{j,\tau}$ in practice. In view of applications in integrable models, we now discuss the case where the Borchers triple has a (tempered) \emph{polarization free generator} in the sense of \cite{BBS:polarizationfree}; see also \cite[Sec.~6]{Lechner:2008}. We emphasize that the modular nuclearity condition is not needed for the following, except where explicitly indicated.

To formulate this, we introduce the Fr\'echet space $\cinfty(\Hil)=\bigcap_{\ell>0} (1+H)^{-\ell}\Hil$ equipped with the norms $\gnorm{ \psi }{}^{(\ell)} := \gnorm{ (1+H)^{\ell}\psi }{} $, cf.~\cite{Bos:short_distance_structure}. We denote by $\lcal(\cinfty(\Hil))$ the linear operators continuous in this topology. $\cinfty(\kcal)$ is defined analogously. 
We now introduce polarization free generators as follows:
\begin{definition}\label{def:pfg}
 A polarization-free generator $\pfg$ assigns to each $\psi\in\cinfty(\kcal)$ a closed, densely defined linear operator $\pfg(\psi)$ on $\Hil$ such that
 \begin{enumerate}[(i)]
  \item \label{it:pfg-domain}
  $\cinfty(\Hil) \subset \dom \pfg(\psi) \cap \dom \pfg(\psi)^\ast$, and $\pfg(\psi)$, $\pfg(\psi)^\ast \in \lcal(\cinfty(\Hil))$ after restriction.
   \item \label{it:pfg-continuous}
For every $\Psi \in \cinfty(\Hil)$, the map $\cinfty(\kcal) \to \cinfty(\Hil)$, $\psi \mapsto \pfg(\psi) \Psi$ is linear and continuous,
  \item \label{it:pfg-vac}
  $\pfg(\psi)\Omega=\psi$ for all $\psi\in\cinfty(\kcal)$,
  \item \label{it:pfg-aff} 
  If $\psi = E_\mu A \Omega$ with some $A =A^\ast \in \M_x'$, then $\pfg(\psi)$ is affiliated with $\mcal_x'$.
 \end{enumerate}

\end{definition}

We will write $\phi(\psi)^\chi$, $\phi(\psi)^\chi_\tau(f)$ analogous to $A^\chi$, $A^\chi_\tau(f)$ earlier; these exists as operators in $\lcal(\cinfty(\Hil))$. 
We now claim

\begin{theorem}\label{thm:pfgscatter}
Suppose that the graded Borchers triple has a mass gap and a polarization-free generator $\pfg$. Let $\psi_1, \ldots,\psi_k \in \kcal$ have compact rapidity support such that $\psi_1 \prec \ldots \prec \psi_k$; and let $\tilde\chi$ be supported near the mass shell with $\tilde\chi(p)=1$ for all $p \in \bigcup_j \supp E(\cdotarg)\psi_j$. Then
\begin{equation}\label{eq:pfgmo}
\mo_\scout (\psi_1 \otimes \dotsm \otimes \psi_k) = \phi(\psi_1)^\chi\dotsm \phi(\psi_k)^\chi \Omega .
 \end{equation}
\end{theorem}
We postpone the proof for a moment. Combining this result with Theorem~\ref{thm:localscatter}, we can conclude:

\begin{corollary}\label{corr:sfrompfg}
Suppose that the graded Borchers triple fulfills modular nuclearity, has a mass gap and a polarization-free generator $\pfg$. Then, for $\psi_1 \prec \ldots \prec \psi_k$ and $\eta_1 \prec \ldots \prec \eta_n$ with compact rapidity support, and $\chi$ as in Theorem~\ref{thm:pfgscatter}, we have
\begin{equation}\label{eq:pfgs}
\begin{aligned}
    \hscalar{ P_\grad \psi_1 \otimes \dotsm \otimes \psi_k & } { S \, P_\grad \eta_n\otimes \dotsm \otimes \eta_1 }_{\fock^\grad(\kcal)}
   \\  &= \frac{1}{\sqrt{k!n!}}
    \hscalar{ \phi(\psi_1)^\chi\dotsm \phi(\psi_k)^\chi \Omega  }{ U(\refl) \phi(U(\refl)\eta_1)^\chi\dotsm \phi(U(\refl) \eta_n)^\chi \Omega }_{\Hil} .
\end{aligned}
\end{equation}
\end{corollary}

Thus in a theory with a polarization-free generator, the scattering matrix is fully determined by matrix elements of the generator. 

In the remainder of this section, we give a proof of Theorem~\ref{thm:pfgscatter}, the hypothesis of which is assumed in the following. We need some technical prerequisites.

\begin{lemma}\label{lemma:swapapprox}
  For any $x,y \in \rbb^2$, $\delta(y-x)>0$, the following inclusions are dense:
  \begin{align}
  \label{eq:swapdense}
     \{\psi : \psi=A\Omega=\hat A\Omega \text{ with some } A \in \M', \hat A \in \M^\twist \} &\subset \Hil, 
\\
  \label{eq:hdense}
     \{\psi :  \psi=A\Omega=\hat A\Omega \text{ with some } A \in \M_y', \hat A \in \M_x^\twist, A,\hat A \in \lcal(\cinfty(\Hil))  \} &\subset \cinfty(\Hil), 
\\
  \label{eq:kdense}
     \{\psi : \psi = E_\mu A\Omega = E_\mu \hat A\Omega \text{ with some } A \in \M_y', \hat A \in \M_x^\twist, A,\hat A \in \lcal(\cinfty(\Hil))  \} &\subset \cinfty(\kcal).
  \end{align}
 Further, given $\psi \in \kcal$ with compact rapidity support, we can approximate it in $\cinfty(\kcal)$ with vectors of the form $E_\mu A(f) \Omega = E_\mu \hat A(f) \Omega$, where $A \in \M_y'$, $\hat A \in \M_x^\twist$, $A,\hat A \in \lcal(\cinfty(\Hil))$, and $\raps{f} \subset \raps{\psi}+[-\epsilon,\epsilon]$, where any $\epsilon>0$ can be chosen.
\end{lemma}

\begin{proof}
  Regarding \eqref{eq:swapdense}, we know that $\{ A\Omega: A \in\M' \}\subset \Hil$ is dense (Def.~\ref{def:triple}(\ref{it:btcyclic})). For finding $\hat A$, it suffices to consider $A$ which is invariant under the involution $A \mapsto Z A^\ast Z^\ast$. Tomita-Takesaki theory yields
  \begin{equation}
      A \Omega = ZA^\ast \Omega = U(\refl) \Delta^{1/2} A \Omega,
  \end{equation}
  hence $\hat A := U(\refl) \Delta^{1/2} A \Delta^{-1/2} U(\refl)\in\M^\twist$ results in $A\Omega=\hat A \Omega$, \emph{supposing} that $\Delta^{1/2} A \Delta^{-1/2}$ exists in $\M'$. A sufficiently large set of such $A$ can be constructed by averaging methods as in \cite[Appendix B]{Duell:wedge_scattering}.

For density in \eqref{eq:hdense}, we may assume $x=-y$. We first note that vectors of the form $\psi(g):=\int d^2x g(x) U(x)\psi$, where $\psi\in\cinfty(\Hil)$, $g \in\ccal_c^\infty(\ocal_{-y,y})$, are dense in $\cinfty(\Hil)$. (Namely, since $U(x)$ is continuous on $\cinfty(\Hil)$, we get $\psi(g) \to \psi$ when $g\to\delta$.) Now given $\psi$ and $g$, and chosen $\epsilon>0$, pick $B \in \M'$, $\hat B \in \M^\twist$ such that $B\Omega=\hat B \Omega$, $\| \psi - B \Omega \| < \epsilon$ as above. Then $A:=B(g)\in\M_{y}'$ and $\hat A :=\hat B(g)\in\M_{-y}^\twist$ are in $\lcal(\cinfty(\Hil))$, as one easily checks; and
  \begin{equation}
      \big\lVert H^\ell \big(\psi(g) - A \Omega\big) \big\rVert = \Big \| \int d^2x \, \partial_0^\ell g(x) \, U(x)(\psi-B\Omega) \Big\| \leq \epsilon \| \partial_0^\ell g \|_1,
  \end{equation}
  which becomes small as $\epsilon \to 0$ at fixed $\ell$. Hence density in \eqref{eq:hdense} holds; then density in \eqref{eq:kdense} follows since $E_\mu$ is continuous in the $\cinfty(\Hil)$-topology. 
  
  Finally, let $\raps{\psi}$ be compact. Let $A,\hat A$ be as in \eqref{eq:kdense}, and choose $f \in \scal(\rbb)$ such that $\tilde f(p)=1$ on $\ucal:=\mu\sinh\raps{\psi}$, $\tilde f=0$ outside an $\epsilon$-neighbourhood $\ucal_\epsilon$ of $\ucal$, and $|\tilde f|\leq 1$ everywhere. Then 
\begin{equation}
   \| A(f) \Omega - \psi \|^2 = \| E_\mu \tilde f(P^1) A \Omega - \psi \|^2 
   \leq \| E_\mu E(p^1 \in \ucal) (A \Omega - \psi) \|^2+ \|E_\mu E(p^1 \in \ucal_\epsilon \backslash \ucal) A \Omega \|.
\end{equation}
The first term can be made small by choice of $A$, and the second as $\epsilon \to 0$ (for given $A$).
\end{proof}

We now show the following variant of \cite[Lemma~3.1(a)]{BBS:polarizationfree}.

\begin{lemma}\label{lemma:pfgwave}
If $\pfg$ is a polarization-free generator, and $\Phi,\hat\Phi\in\cinfty(\Hil)$, then
\begin{equation}\
   g_{\Phi,\hat\Phi,\psi} :  \rbb^2 \to \cbb, \quad x \mapsto \hscalar{\Phi}{U(x) \pfg(\psi) U(x)^\ast \hat\Phi}
\end{equation}
is a smooth solution of the Klein-Gordon equation with mass $\mu$. If $f\in\scal(\rbb^2)$ such that $\raps{f}$ is compact, and $\tilde \chi$ has support near the mass shell, then  
$\pfg(\psi)^\chi_\tau(f)$ is independent of $\tau$.
\end{lemma}

\begin{proof}
 Using Lemma~\ref{lemma:swapapprox}, the continuity properties of $\pfg$, and the fact that the generators of $U(x)$ are continuous on $\cinfty(\Hil)$ (hence limits of solutions are solutions), it suffices to consider the case $\hat\Phi = A\Omega$ and $\psi=E_\mu B \Omega$ with $A \in \M_y\cap \lcal(\cinfty(\Hil))$ for some $y\in\wedg_R$, and $B \in \M'_y\cap \lcal(\cinfty(\Hil))$. 
 In this case, since $\pfg(\psi)$ is affiliated with $\M_y'$, and for $\delta(y-x)>0$,
 \begin{equation}
   g_{\Phi,\hat\Phi,\psi} (x) =  \hscalar{A^\ast \Phi}{U(x) \pfg(\psi) U(x)^\ast \Omega} = \hscalar{A^\ast \Phi}{U(x) E_\mu B \Omega}. 
 \end{equation}
 This is a solution of the Klein-Gordon equation since $U(x)E_\mu$ is. Thus $g_{\Phi,\hat\Phi,\psi}$ is a solution at least in the region $x \in \wedg_L+y$. Since $y$ was arbitrary, the first claim follows.
 
 Now one computes for $\Phi,\hat\Phi \in \cinfty(\Hil)$,
 \begin{equation}
    \hscalar{\Phi}{ \pfg(\psi)^\chi_\tau(f) \hat\Phi}
    = \int dx\, f(\tau,x)\, \big(\chi \ast g_{\Phi,\hat\Phi,\psi}\big)(\tau,x)   
 \end{equation}
 where $\ast$ denotes convolution in two dimensions. By the support properties of $\tilde \chi$, the function $\chi \ast g_{\Phi,\hat\Phi,\psi}$ is a negative-energy solution of the Klein-Gordon equation, while $(\tau,x)\mapsto f(\tau,x)$ is a positive-energy solution by its definition in \eqref{eq:achidef}, hence the integral expression is independent of $\tau$. 
\end{proof}

Next, we note the following improvement of the convergence statement    in Theorem~\ref{thm:wedgescatter}:
\begin{proposition}\label{prop:improvedconv}
Let $x,y\in\rbb^2$ and $A_1,\ldots,A_n \in \M_y'$, $\hat A_2,\ldots,\hat A_n \in \M_x^\twist$ such that $E_\mu A_j \Omega = E_\mu \hat A_j \Omega$ for $j\geq 2$. 
Let $f_1,\ldots,f_n\in\scal(\rbb)$ such that $\raps{f_1} \prec \dots \prec \raps{f_n}$. 
For every $\ell\geq 0$, $N >0$ there exists a constant $c$ such that for all $\tau \geq 1$,
\begin{equation}
\begin{split}
     \| A_{1,\tau}^\chi(f_1) A^\chi_{2,\tau}(f_2)\dotsm A^\chi_{n,\tau}(f_n)\Omega 
     - \mo_\scout \big(A^\chi_{1}(f_1)\Omega \otimes \dotsm \otimes A^\chi_{n}(f_n)\Omega\big) \|^{(\ell)} 
    \\ \leq c \tau^{-N} \|A_1\| \prod_{j\geq 2} \big( \|A_j\| + \|\hat A_j\| \big)
\end{split}
\end{equation}
where the constant $c$ may depend on $n$, $f_j$, $\chi$, $\ell$, $N$ but not on the $A_j$, $\hat A_j$.
\end{proposition}

\begin{proof}
 We indicate the changes over \cite[Theorem~6]{Duell:wedge_scattering}, referring mostly to the notation and numbering used in that paper:
 \begin{enumerate}[(a)]
 \item As evident from the proof of Theorem 6 there, ``swapability'' of the vector $E_\mu A_j \Omega $ is only required for $j \geq 2$; the base case of the induction does not require it. Hence we can drop the operator $\hat A_1$ from our assumptions.
 \item All estimates there, in particular in Proposition~8(iv),(v), hence in Lemma 9 and Corollary 10, can be done uniformly in $\|A_j\|, \|\hat A_j\|$; see also \cite[Appendix~B]{Duell:reehschlieder}.
  \item Cook's method in Eq.~(30) can likewise be applied to bound $\|H^\ell \Psi_n(\tau) \|$, by estimating $\|\partial_\tau^{\ell+1} \Psi_n(\tau) \| \leq C_{N,\ell} \tau^{-N}$. Note that, since $f_j(\tau,x)$ solve the Klein-Gordon equation, only first time derivatives of $f$ actually occur in the estimates. 
  \item The graded case can be accommodated by showing Corollary 10 analogously for even and for odd operators $B$, $B^\perp$, then splitting all $A_j$, $\hat A_j$  into even and odd parts and replacing their commutators with graded expressions as in \eqref{eq:gradpermute}. \qedhere
\end{enumerate}
\end{proof}

Finally, we make use of the following approximation property for closed operators affiliated with a von Neumann algebra.\footnote{The authors would like to thank Y.~Tanimoto for discussions on this aspect.}

\begin{lemma}\label{lemma:closeapprox}
Let $T$ be a closed, densely defined operator affiliated with a von Neumann algebra $\ncal$. Then for each $\epsilon>0$ there exists $C_\epsilon\in\ncal$ such that $\|C_\epsilon\| \leq \epsilon^{-1}$ and for each $\Psi \in \dom T^\ast T$,
\begin{equation}
 \| (T-C_\epsilon) \Psi\| \leq  \epsilon \,\|T^\ast T \Psi\|.  
\end{equation}
\end{lemma}
\begin{proof}
  With $T = V |T|$ the polar decomposition, set $C_\epsilon := V \epsilon^{-1} \sin(\epsilon |T|)\in\ncal$. Then clearly $\|C_\epsilon\|\leq\epsilon^{-1}$ and
  \begin{equation}
    \| (T-C_\epsilon) \Psi\| \leq  \epsilon \| g(\epsilon |T|) \, |T|^2 \, \Psi \|  \quad \text{where } g(\lambda) =  \frac{\lambda - \sin \lambda}{\lambda^2}. 
  \end{equation}
  It now suffices to note that $|g(\lambda)| \leq 1$.
  \end{proof}

\begin{proof}[Proof of Theorem~\ref{thm:pfgscatter}]
By Lemma~\ref{lemma:swapapprox} and the continuity properties of polarization free generators, we can assume that $\psi_j= E_\mu A_j (f_j)\Omega =E_\mu \hat A_j (f_j)\Omega$ with $A_j \in \M_y' \cap \lcal(\cinfty(\Hil))$, $\hat A_j \in \M_x^\twist \cap \lcal(\cinfty(\Hil))$, and $f_1 \prec \dots \prec f_k$. 
By linearity, we may also assume that $A_1=A_1^\ast$ (noting that $\hat A_1$ does not play a role in the following).
We will show that
\begin{equation}\label{eq:morecursive}
  \mo_\scout (\psi_1 \otimes \dotsm \otimes \psi_k) = \pfg(\psi_1)^\chi \mo_\scout (\psi_2 \otimes \dotsm \otimes \psi_k) =: \pfg(\psi_1)^\chi \Psi_{k-1} 
\end{equation}
of which the claim \eqref{eq:pfgmo} follows by induction on $k$. 

For that, first note that the expressions in \eqref{eq:morecursive} are well-defined: With the restrictions introduced above, we know that $\psi_j\in\cinfty(\kcal)$. Further, due to proprety \eqref{eq:mosymm} for time translations, $\mo_\scout$ intertwines $H$ with the second-quantized Hamiltonian on $\fock(\kcal)$. Together, this implies that $\Psi_{k-1} \in \cinfty(\hcal)$; in particular, $\Psi_{k-1}$ is in the domain of $\pfg(\psi_1)^\chi$.

Now using that $\pfg(E_\mu A_1 \Omega)$ is affiliated with $\M_y'$, choose an approximating sequence $C_\epsilon\in\M_y'$ for $\pfg(E_\mu A_1 \Omega)$ as in Lemma~\ref{lemma:closeapprox}. 
Then  
 \begin{equation}\label{eq:pfgest1}
\big\lVert \pfg(\psi_1)^\chi(f_1) \Psi_{k-1} -C_{\epsilon,\tau}^\chi(f_{1}) \Psi_{k-1} \big\rVert 
\leq  \epsilon \, \|\chi\|_1 \,\|f_{1}(\tau,\cdotarg)\|_1\,
 \big\lVert \,\lvert \pfg(E_\mu A_1 \Omega)\rvert^2 \Psi_{k-1} \big\rVert 
 \leq 
  c \epsilon \tau   ^{1/2}  
 \end{equation}
 using the standard estimate  $\|f_{1}(\tau,\cdotarg)\|_1 = O(\tau^{1/2})$, and noting that $\pfg(\psi_1)^\chi_\tau(f_{1})$ is independent of $\tau$ by  Lemma~\ref{lemma:pfgwave}. 
 
 Further, using Proposition~\ref{prop:improvedconv} (for $k-1$ in place of $k$), we find
 \begin{equation}\label{eq:pfgest2}
 \begin{aligned}
    \big\lVert C_{\epsilon,\tau}^\chi(f_{1})  A_{2,\tau}^\chi(f_2) \dotsm A_{k,\tau}^\chi(f_k)\Omega & - C_{\epsilon,\tau}^\chi(f_{1}) \Psi_{k-1} \big\rVert
  \\
    &\leq \| C_{\epsilon}\|\, \|\chi\|_1 \, \|f_{1}(\tau,\cdotarg)\|_1\, c' \tau^{-N}
    \leq c'' \epsilon^{-1} \tau^{1/2} \tau^{-N}
\end{aligned}
  \end{equation}
 with constants $c'$, $c''$, where $N$ can be chosen arbitrary. Further, from Proposition~\ref{prop:improvedconv},
\begin{equation}\label{eq:pfgest3}
    \Big\lVert C_{\epsilon,\tau}^\chi(f_{1})  A_{2,\tau}^\chi(f_2) \dotsm A_{n,\tau}^\chi(f_n) \Omega- \mo_\scout\big(E_\mu C_\epsilon^\chi(f_1) \Omega \otimes \psi_2\otimes\dotsm \otimes \psi_n \big)  \Big\rVert
    \leq c''' \epsilon^{-1} \tau^{-N}.
\end{equation}

The final step is as follows: We note that
\begin{equation}
   \psi_1 = \pfg(\psi_1)\Omega = \pfg(\psi_1)^\chi \Omega =   \pfg(E_\mu A_1 \Omega)^\chi \Omega,
\end{equation}
where the first equality follows since $\pfg$ is a polarization-free generator, and the second since $\tilde\chi(p)=1$ for all $p \in \supp E(\cdotarg) \psi_1$. Now using isometry of $\mo_\scout$ on $\fock^\prec(\kcal)$, we find
\begin{equation}\label{eq:pfgest4}
\begin{aligned}
    \Big\lVert \mo_\scout\big(E_\mu C_\epsilon^\chi(f_1)\Omega\otimes \psi_2 \otimes \dotsm \otimes \psi_n\big) &- \mo_\scout(\psi_1 \otimes \dotsm \otimes \psi_n)  \Big\rVert
    \\ &\leq \, 
    \big\lVert E_\mu \big(C_\epsilon-\pfg(E_\mu A_1 \Omega)\big)^\chi(f_1)\Omega \big\rVert \, \| \psi_2 \otimes \dotsm \otimes \psi_n  \|
    \\ &\leq c'''' \epsilon \tau^{1/2} \tau^{-N}. 
\end{aligned}
\end{equation}
Overall, choosing $\epsilon = \tau^{-1}$ and $N$ large enough, all of \eqref{eq:pfgest1}, \eqref{eq:pfgest2}, \eqref{eq:pfgest3}, \eqref{eq:pfgest4} vanish as $\tau \to \infty$, showing the claim. 
 \end{proof}

\section{Graded integrable models} \label{sec:integrable}

We now present a class of examples for graded Borchers triples, taken from the theory of integrable models. Specifically, we take the Borchers triples for integrable models with one species of scalar particle \cite{Lechner:2008}, which we now consider with a grading.

\subsection{Wedge-local fields and algebras}

We recall how these triples were constructed: Let $S$ be a meromorphic function on $\cbb$ which is analytic and bounded on the strip $\rbb + i[0,\pi]$ and fulfils the symmetry conditions
\begin{equation}\label{eq:srelat}
 S(\zeta)^{-1}=S(-\zeta)=\overline{S(\bar{\zeta})\vphantom{\hat S}}=S(\zeta+i \pi).
\end{equation}
Related to this, we consider generalized annihilation and creation operators with rapidity-dependent kernels $z(\theta)$, $\zd(\theta)$, which fulfill the Zamolodchikov relations,
\begin{equation}\label{eq:zamol}
\begin{aligned}
\zd(\theta)\zd(\theta') &= S(\theta-\theta')\zd(\theta')\zd(\theta),\\
z(\eta)z(\eta') &= S(\eta-\eta')z(\eta')z(\eta),\\
z(\eta)\zd(\theta) &= S(\theta-\eta)\zd(\theta)z(\eta) + \delta(\theta-\eta) \idop.
\end{aligned}
\end{equation}
They act on an ``$S$-symmetrized'' Fock space $\Hil$ over the single-particle space $\kcal=L^2(\rbb,d\theta)$. On this space, one has a second-quantized version of $\pcal_+^\uparrow$, in particular the translations $U(x)$ of the Borchers triple, and the Fock vacuum vector $\Omega$.

The algebra $\M$ is defined as follows. We consider the operator-valued distribution,
\begin{equation}\label{eq:wlfield}
\pfg(f):=  \int d\theta \, \Big(\tilde f(p(\theta)) \zd(\theta) +  \tilde f(-p(\theta)) z(\theta) \Big) , \qquad 
 f\in \mathcal{S}_\rbb(\mathbb{R}^{2}),
\end{equation}
where $p(\theta)=\mu(\cosh\theta,\sinh\theta)$. This $\pfg(f)$, formally analogous to a free Bose field, is $\pcal_+^\uparrow$-covariant but not local in the usual sense. It is essentially selfadjoint, hence we can define
\begin{equation}
 \M := \{  \exp i \pfg(f)^- | \supp f \subset \wedg_L \}'.
\end{equation}
Then $(\M,U,\Omega)$ fulfills the usual axioms of a Borchers triple \cite{Lechner:2008}. It is worth noting that $\M$ can also be obtained as
\begin{equation}\label{eq:}
 \M = \{  \exp i \phi'(f)^- | \supp f \subset \wedg_R \}''
\end{equation}
with the ``right field'' $\phi'(f) = J \pfg(\refl.f) J$ where $J$ is the reflection operator of the (non-graded) Borchers triple.

We now add the grading $\grad = (-1)^N$ with $N$ denoting the particle number operator. Since $\grad$ anticommutes with the $z$, $\zd$, we have $\grad \pfg(f) \grad = -\pfg(f)$, and hence $\grad \mcal \grad = \mcal$. Also, $\grad$  commutes with all $U(g)$, $g \in \pcal_+^\uparrow$, and leaves $\Omega$ invariant, hence $(\M,U,\Omega,\grad)$ is a graded Borchers triple in the sense of Definition~\ref{def:triple}.

With the twist induced by this $\grad$, we note that
\begin{equation}
\M^\twist= \{ \exp i \hat\pfg(f)^- | \supp f \subset \wedg_R \}'' \quad \text{where } \hat\pfg(f) = Z \phi'(f) Z^\ast  = i (-1)^N \phi'(f).
\end{equation}
That is, compared with the situation in \cite{Lechner:2008}, we ``twist'' the right field (replace $\phi'$ with $\hat \phi$) but leave the left field invariant ($\phi$ is unchanged).

For a certain class of $S$, the modular nuclearity condition has been established, including the case $S=1$ \cite{BDL:modular_nuclearity} and those with a bounded extension to a strip $-\epsilon < \im \zeta < \pi + \epsilon$ and $S(0)=-1$ \cite{AL2017}. Hence at least in these cases, Theorem~\ref{thm:gradednet} applies, and our double cone algebras $\F(\ocal_{x,y})$ are nontrivial for suitably large double cones. Thus the construction yields a nontrivial, twisted-local net.

In terms of scattering theory, we note that $\pfg$ is indeed a polarization-free generator in the sense of Definition~\ref{def:pfg}: Extending it by $\phi(\psi):=\zd(\psi)+z(\bar \psi)$ to all $\psi \in \kcal$, then passing to its operator closure, the domain and continuity conditions in Def.~\ref{def:pfg}(\ref{it:pfg-domain}),(\ref{it:pfg-continuous}) follow from the $H$-bound $\|(1+H)^{-1/2} \phi(\psi)\|<\infty$ and from the relation $[(1  +H)^{-1},\phi(\psi)] = i(1+H)^{-1} \phi(-iH \psi) (1+H)^{-1}$ (cf.~\cite{FreHer:pointlike_fields}). Further if $\psi = E_\mu A \Omega$, $A =A^\ast\in \mcal'$, meaning that $\psi$ lies in the usual ``standard subspace'' associated with $\wcal_L$ in free field theory, one finds as in \cite[p.~151]{Lechner:2003} that $\phi(\psi)$ commutes spectrally with all $\phi'(g)$, $\supp g \subset \wcal_R$, and hence $\phi(\psi)$ is affiliated with $\M'$.

Also, if $\raps{\psi}$ is compact, $\tilde \chi$ has support near the mass shell and value $1$ on $\supp \tilde\psi$, we have $\phi(\psi)^\chi=\zd(\psi)$. Therefore, at least in cases with modular nuclearity, we obtain from Corollary~\ref{corr:sfrompfg},
\begin{equation}
\begin{aligned}
    \hscalar{ P_\grad \psi_1 \otimes \dotsm \otimes \psi_k } { S &\, P_\grad \eta_n\otimes \dotsm \otimes \eta_n }_{\fock^\grad(\kcal)}
     \\ &= \frac{i^n}{\sqrt{k!n!}}
    \hscalar{ \zd(\psi_{\sigma(1)})\dotsm \zd(\psi_{\sigma(k)}) \Omega  }{ U(\refl) \zd(\eta_{\rho(1)})\dotsm \zd( \eta_{\rho(n)}) \Omega }_{\Hil},
\end{aligned}
\end{equation}
when $\sigma$, $\rho$ are the permutations that ``order'' the rapidity support of the $\psi_j$ and $\eta_j$ respectively. We see that the above expression is nonzero only for $k=n$, and using the form of $U(\refl)$ and \eqref{eq:zamol}, we find for the integral kernel of $P_\grad S P_\grad$ at $n$-particle level 
\begin{equation}
 (P_\grad S P_\grad)_n (\thetav,\etav) =   \Big(\frac{1}{n!}\sum_{\sigma\in\perms{n}} \prod_{j=1}^n\delta(\theta_j-\eta_{\sigma(j)}) \Big)
 \prod_{k < m} \Big(-S(|\theta_k-\theta_m|) \Big).
\end{equation}
That is, the model has a factorizing scattering matrix with asymptotic \emph{fermions} (the space $\fock^\grad$ is the antisymmetric Fock space) and a 2-particle scattering matrix $-S$ rather than $S$. (Cf.~\cite[Theorem~6.3]{Lechner:2008} for the bosonic case.)
The simplest interacting example is certainly $S=1$, where the 2-particle scattering matrix is $-S(\theta)=-1$; a fermionic analogue to the (bosonic) massive Ising model.

For illustration, we also note the special case $S=-1$: In this case, the relations \eqref{eq:zamol} specialize to the CAR. Instead of the field $\pfg$ in \eqref{eq:wlfield}, one may consider a free Majorana field, with components
\begin{equation}
   \psi_{\pm}(f) = \int d\theta \,   e^{i \pi(-2\pm 1)/4} e^{\pm\theta/2} \tilde f(p(\theta)) \zd(\theta) + h.c.
\end{equation}
One verifies that $\psi_\pm(f)$, $\supp f \subset \wedg_L$ commutes with $\phi'(g)$, $\supp g \subset \wedg_R$ (note that both are bounded operators in this case),
and in fact, using that $U(\refl) \psi_\pm(f) U(\refl) = \psi_\pm(\refl.f)$, one can check that
\begin{equation}
  \M' = \{ \psi_\pm(f) : \supp f \subset \wedg_L \}'', \quad
  \M^\twist = \{ \psi_\pm(f) : \supp f \subset \wedg_R \}''.
\end{equation}
In other words, our net $\F$ is generated by a free Majorana field; this is consistent with our result that the asymptotic states are fermions and the 2-particle scattering matrix is $-S(\theta)=1$.

\subsection{Connection to the form factor programme}

We now want to analyze the structure of operators in the local field algebras $\F(\ocal_{x,y})$, exhibiting a connection to the form factor approach to integrable models. This is largely analogous to the analysis in \cite{BostelmannCadamuro:characterization} for the bosonic case, hence we shall mainly confine ourselves to giving a rough overview and to pointing out the necessary changes in the fermionic case.

First off, independent of the structure of the local operators, every operator on the $S$-symmetric Fock space $\Hil$ can be expanded \cite[Thm.~3.8]{BostelmannCadamuro:expansion} in a series of annihilators and creators,
\begin{equation}\label{eq:fmnseries}
   A = \sum_{m,n} \int \frac{d^m \thetav d^n \etav}{m!n!} \cme{m,n}{A}(\thetav,\etav) z^{\dagger}(\theta_1)\dotsm  z^\dagger(\theta_m)z(\eta_1)\dotsm z(\eta_n),
\end{equation}
where the coefficients $\cme{m,n}{A}(\thetav,\etav)$ depend linearly on $A$. More precisely, the expansion holds in the following sense. We choose an \emph{indicatrix} $\omega$, a subadditive, less-than-linearly growing function (see \cite[Def.~2.1]{BostelmannCadamuro:characterization} for details). Based on this, we consider a class of quadratic forms $A\in\qf^\omega$ on $\Hil$, essentially those where $Ae^{-\omega(H)}$ and $e^{-\omega(H)} A$ are bounded at fixed particle number cutoff on both sides. (In particular, $\qf^\omega$ includes all bounded operators.) Then the expansion holds for all $A\in\qf^\omega$, with the $\cme{m,n}{A}$ being distributions of a certain regularity class. (See \cite[Sec.~2.3]{BostelmannCadamuro:characterization}.) The series is read in matrix elements with particle number cutoffs, where it is a finite sum; convergence issues do not arise.

Moreover, we extend the notion of locality from bounded operators to quadratic forms. We say that $A\in\qf^\omega$ is $\omega$-local in the wedge $\wedg_L$ if it commutes with $\phi'(g)$ for all $\supp g \subset \wedg_R$ in the weak sense; in $\wedg_R$ if $U(\refl) A U(\refl)$ is $\omega$-local in $\wedg_R$; and analogously for shifted wedges; finally, in $\ocal_{x,y}$ if it is $\omega$-local in both $\wedg_R+x$ and in $\wedg_L+y$. (See \cite[Def.~2.4]{BostelmannCadamuro:characterization}.)

Noting that the notion of localization in $\wedg_L$ is independent of the grading, the following characterization from \cite{BostelmannCadamuro:characterization} applies identically in our context:

\begin{theorem}\label{thm:ffwedge}
$A \in \qcal^\omega$ is $\omega$-local in $\wedg_L$ if and only if 
\begin{equation}\label{eq:fmnfromF}
  \cme{m,n}{A}(\thetav,\etav)=F_{m+n}(\thetav+i\zerov,\etav+i\piv-i\zerov).
\end{equation}
where the functions $F_k$ ($k \in \nbb_0$) fulfill:

\begin{enumerate}
\renewcommand{\theenumi}{(FW\arabic{enumi})}
\renewcommand{\labelenumi}{\theenumi}

\item \label{it:fwmero}
$F_k$ is analytic on $\rbb^k + i \ical^k_+$, where $\ical^k_+:=\{\lambdav: 0 < \im \lambda_1 < \dots < \im\lambda_k < \pi\}$,

\item \label{it:fwsymm} 
For any $\sigma \in \perms{k}$, we have
$
\displaystyle{
F_k(\thetav+i\zerov)
= S^\sigma(\thetav) F_k(\thetav^\sigma +i\zerov)
}
$, where $\thetav^\sigma =(\theta_{\sigma(1)},\ldots,\theta_{\sigma(k)})$.

\item \label{it:fwboundsreal}
For each $j \in \{0,\ldots,k\}$, we have $\| F_k\big( \cdotarg + i (0,\dotsc,0,\underbrace{\pi,\dotsc,\pi}_{\text{$j$ entries}}) + i \zerov\big) \|_{(k-j) \times j}^{\omega} < \infty$.

\item \label{it:fwboundsimag}
%
There exist $c,c'>0$ such that for all $\im\zetav\in\ical^k_+$,
\begin{equation*}
  |F_k(\zetav)| \leq c \operatorname{dist}(\im \zetav,\partial \ical_+^k)^{-k/2} \prod_{j=1}^k \exp \big(\mu r  \im \sinh \zeta_j+ c' \omega(\cosh \re \zeta_j)\big).
\end{equation*}
\end{enumerate}
\noindent
\end{theorem}
Here $+i\zerov$ denotes the boundary distribution when approached from within $\ical_+^k$. The factor $S^\sigma$ is as in 
\cite[Eq.~(2.2)]{BostelmannCadamuro:expansion}. 
For the norm \mbox{$\|\cdotarg \|_{(k-j) \times j}^{\omega}$} see \cite[Eq.~(2.25)]{BostelmannCadamuro:characterization}, but its details will not matter here. 

Corresponding statements hold for operators localized in translated wedges, with $\hat A = U(x)AU(x)^\ast$ corresponding to the analytic functions $\hat F_k (\zetav)= e^{ip(\zetav)\cdot x}F_k(\zetav)$.

The notion of localization in $\wedg_R$, however, and hence in double cones, does depend on the grading, if only via the form of the reflection operator, $U(\refl)=ZJ$. We focus on locality in the standard double cone $\ocal_r$, $r>0$, and obtain by analogy with \cite[Theorem 5.4]{BostelmannCadamuro:characterization}:

\begin{theorem}\label{thm:ffdcone}
$A \in \qcal^\omega$ is $\omega$-local in $\ocal_{r}$ if and only if
\begin{equation}
  \cme{m,n}{A}(\thetav,\etav)=F_{m+n}(\thetav+i\zerov,\etav+i\piv-i\zerov).
\end{equation}
where the functions $F_k$ ($k \in \nbb_0$) fulfill:

\begin{enumerate}
\renewcommand{\theenumi}{(FD\arabic{enumi})}
\renewcommand{\labelenumi}{\theenumi}

\item \label{it:fdmero}
$F_k$ is meromorphic on $\cbb^k$, and analytic where $\im \zeta_1 < \ldots < \im \zeta_k < \im \zeta_1 + \pi$.

\item \label{it:fdsymm} 
For any $\sigma \in \perms{k}$, we have
$
\displaystyle{
F_k(\zetav)
=  S^\sigma(\zetav) F_k(\zetav^\sigma) .
}
$

\renewcommand{\theenumi}{(FD\arabic{enumi}$_\grad$)}

\item \label{it:fdperiod} 
$\displaystyle{
F_k (\zetav + 2i\pi \ev^{(j)} ) =
{(-1)^k} \Big(\prod_{\substack{i=1 \\ i \neq j}}^k S(\zeta_i-\zeta_j)\Big)  F_k (\zetav ) .
}$

\item \label{it:fdrecursion}

%
The $F_k$ have first order poles at $\zeta_n-\zeta_m = i \pi$, where $1 \leq m < n \leq k$,
and one has with $\hat\zetav = (\zeta_1,\ldots,\widehat{\zeta_m},\ldots, \widehat{\zeta_n}, \ldots, \zeta_k)$ where the hats denote omitted variables,
\begin{equation*}
\res_{\zeta_n-\zeta_m = i \pi} F_{k}(\boldsymbol{\zeta})
= - \frac{1}{2\pi i }
\Big(\prod_{j=m}^{n} S(\zeta_j-\zeta_m) \Big)
\Big(1- {(-1)^k}\prod_{p=1}^{k} S(\zeta_m-\zeta_p) \Big)
F_{k-2}( \boldsymbol{\hat\zeta} ).
\end{equation*}

\renewcommand{\theenumi}{(FD\arabic{enumi})}

\item \label{it:fdboundsreal}
For each $j \in \{0,\ldots,k\}$, we have
\begin{equation*}
\onorm{ F_k\big( \cdotarg + i (0,\dotsc,0,\underbrace{\pi,\dotsc,\pi}_{\text{$j$ entries}}) + i \zerov \big) }{(k-j) \times j} < \infty,
\quad
\onorm{ F_k\big( \cdotarg + i (-\pi,\dotsc,-\pi,\underbrace{0,\dotsc,0}_{\text{$j$ entries}}) + i \zerov \big) }{(k-j) \times j} < \infty.
\end{equation*}
Here $+i\zerov$ denotes approach from inside the region of analyticity as in \ref{it:fdmero}.

\item \label{it:fdboundsimag}
%
There exist $c,c'>0$ such that for all $\im \zetav\in\ical^k_\pm$, where $\ical^k_- = \ical^k_+ - (\pi,\dotsc,\pi)$: 
\begin{equation*}
  |F_k(\zetav)| \leq c \,{ \operatorname{dist}(\im \zetav,\partial \ical_\pm^k)^{-k/2}} \prod_{j=1}^k \exp \big(\mu r  |\im \sinh \zeta_j|+ c' \omega(\cosh \re \zeta_j)\big).
\end{equation*}
\end{enumerate}

\end{theorem}

This is in line with what is expected by the form factor programme \cite[Sec.~3]{Lashkevich:1994}, noting some change in conventions. 

\begin{proof}
We outline the proof, pointing out the necessary changes to the arguments in \cite[Sec.~5]{BostelmannCadamuro:characterization}.
If $A$ is $\omega$-local in the double cone, then both to $U(-r \ev^{(1)})AU(-r\ev^{(1)})^\ast$ and  $U(-r\ev^{(1)})U(\refl)AU(\refl)U(-r\ev^{(1)})^\ast$, Theorem~\ref{thm:ffwedge} can be applied. On the other hand, we have the key formula
\begin{equation}\label{eq:fmnrefl}
\cme{m,n}{U(\refl)A U(\refl)}(\thetav, \etav) =(-1)^{m(m+n)} \sqrt{(-1)^{m+n}} \sum_{D \in \mathcal{C}_{m,n}} (-1)^{|D|(m+n-1)}  \delta_D S_D \dot{R}_D(\thetav, \etav) f^{[A]}_{n- |D|, m - |D|}(\hat{\etav}, \hat{\thetav}),
\end{equation}
which is altered from \cite[Prop.~3.11]{BostelmannCadamuro:expansion} to take the modified reflection operator $U(\refl)$ into account. Here the sum runs over all contractions $D$ of $m+n$ indices (see \cite{BostelmannCadamuro:expansion} for this notation, including the hatted variables), the factors $\delta_D$, $S_D$ are as in \cite{BostelmannCadamuro:expansion}, and $\dot R_D$ denotes the factor $R_D$ from \cite[Eq.~(3.44)]{BostelmannCadamuro:expansion} with $S$ replaced by $-S$; the square root has values in the upper half plane.

Therefore, we can define an analytic function $F_k$ on $\rbb^k + i \ical^k_\pm$ in terms of those arising from Theorem~\ref{thm:ffwedge} by 
\begin{equation}
  F_k (\zetav) :=\begin{cases}
    F_k^{[A]}(\zetav) &\text{for }\im\zetav\in\ical^k_+,   
\\
    \sqrt{(-1)^k}\, F_k^{[U(\refl)AU(\refl)]}(\zetav +i \piv) &\text{for }\im\zetav\in\ical^k_-,   
                \end{cases}
\end{equation}
which match at the real hyperplane due to \eqref{eq:fmnrefl}. (Cf.~\cite[Eq.~(5.7)]{BostelmannCadamuro:characterization}.) This automatically has properties \ref{it:fdboundsreal} and \ref{it:fdboundsimag}. We can extend the function periodically by 
\begin{equation}
 F_k  (\zetav+2i\piv):= (-1)^k  F_k (\zetav) 
\end{equation}
where again the boundaries match due to \eqref{eq:fmnrefl}. (Cf.~\cite[Eq.~(5.13)]{BostelmannCadamuro:characterization}.) By extension to the convex hull using the tube theorem, this defines $F$ on the domain $\im \zeta_1 < \dots < \im\zeta_k < \im \zeta_1+\pi$.
Further, one can extend $F$ \emph{meromorphically} to the domain $\im \zeta_1 < \dots < \im\zeta_k < \im \zeta_1+2\pi$ by setting
\begin{equation}\label{eq:stairs}
F_k(\zetav):= (-1)^{mk} F_{k}(\zeta_{m+1},\ldots,\zeta_{k},\zeta_{1}+2i\pi,\ldots,\zeta_{m}+2i\pi)
\end{equation}
for all $0 < m \leq k$, and again using the tube theorem. (Cf.~\cite[Eq.~(5.15)]{BostelmannCadamuro:characterization}.)
From differences of boundary values \eqref{eq:fmnrefl}, one computes that $F_k$ has a first-order
pole at $\zeta_i-\zeta_j=i\pi$, with residue given by \ref{it:fdrecursion}. $F$ can then be further continued as a meromorphic function on $\cbb^k$ using \ref{it:fdsymm} as a definition, which is consistent on the existing domain due to \ref{it:fwsymm}. Now \ref{it:fdperiod} follows by combining \eqref{eq:stairs} with \ref{it:fdsymm}.

Conversely, let $F_k$ be meromorphic functions fulfilling conditions (FD$_\grad$). Defining $A$ by the expansion \eqref{eq:fmnseries}, with $\cme{m,n}{A}(\thetav,\etav)=F_{m+n}(\thetav + i\zerov,\etav+i\piv-i\zerov)$, it is immediate from Theorem~\ref{thm:ffwedge} that this $A$ is $\omega$-local in $\wedg_L+r\ev^{(1)}$. Likewise, the quadratic form $A^\pi$ defined by the expansion coefficients $\cme{m,n}{A^\pi}(\thetav,\etav)=\sqrt{(-1)^{m+n}}F_{m+n}(\thetav -i\piv+ i\zerov,\etav-i\zerov)$ is $\omega$-local in $\wedg_L+r\ev^{(1)}$. However, using Eq.~\eqref{eq:fmnrefl} and conditions (FD$_\Gamma$), one shows that
\begin{equation}
  \cme{m,n}{U(\refl)A^\ast U(\refl)} = \cme{m,n}{A^\pi}\, \text{ for all $m,n$};
\end{equation}
in other words, $U(\refl)A^\ast U(\refl) = A^\pi$.
(Cf.~\cite[Prop.~5.10]{BostelmannCadamuro:characterization}.)
Thus $A = U(\refl)(A^\pi)^\ast U(\refl)$ is also $\omega$-localized in $\refl . (\wedg_L+r\ev^{(1)})= \wedg_R-r\ev^{(1)}$, which concludes the proof.
\end{proof}

This characterizes the elements of $\F(\ocal_r)$, which under nuclearity assumptions we know to exist by abstract reasons. 
However, we can also give explicit examples of solutions of the form factor equations. Let us consider the example $S=1$---which, due to the presence of the grading $\Gamma=(-1)^N$, does \emph{not} yield the net of the free Bose field. 
Picking a test function $g$ supported in $\ocal_r$, the following sequence of analytic functions fulfills the conditions in Theorem~\ref{thm:ffdcone}:
\begin{equation}
F_{2k}=0, \quad F_{2k+1} = \frac{(-1)^k}{(4\pi)^k k!}\,\tilde g(p(\zetav)) \sum_{\sigma \in \perms{2k+1}} e^{\mp \frac{1}{2} \theta_{\sigma(1)}} \prod_{i=1}^k \operatorname{sech} \frac{\theta_{\sigma(2i)} - \theta_{\sigma(2i+1)}}{2}.
\end{equation}
(Cf.~\cite[Eq.~(4.5)]{Lashkevich:1994}.) This hence gives a quadratic form $A \in \qf^\omega$ which is $\omega$-local in $\ocal_r$. It is less imminent whether the form extends to a closed or even bounded operator. However, with the summation methods of \cite[Sec.~5]{BostelmannCadamuro:examples}, it should be possible to establish $A$ as a closable operator for suitable choices of $g$ and $\omega$.

It is natural to ask whether the disorder operators $V$ can be characterized in terms of their coefficients $\cme{m,n}{V}$. Since unitarity and indeed boundedness of $V$ is very difficult to encode into the series expansion \eqref{eq:fmnseries} that involves unbounded objects, we will formulate this result only in one direction: we derive analyticity properties of $\cme{m,n}{V}$ when $V \in \diso_{L}(\ocal_r)$. For a full characterization, one would need to pass to possibly unbounded disorder-type operators, defined in terms of the (weak) commutation relations with the wedge-local fields, a step that we skip here for brevity. We claim:

\begin{theorem}\label{thm:ffdiso}
If $V \in \diso_{L}(\ocal_r)$ , then
\begin{equation}
  \cme{m,n}{V}(\thetav,\etav)=F_{m+n}(\thetav+i\zerov,\etav+i\piv-i\zerov).
\end{equation}
where the functions $F_k$ ($k \in \nbb_0$) fulfill conditions \ref{it:fdmero}, \ref{it:fdsymm}, \ref{it:fdboundsreal}, \ref{it:fdboundsimag} from Theorem~\ref{thm:ffdcone}, which vanish for odd $k$, and:

\begin{enumerate}
\setcounter{enumi}{2}
\renewcommand{\theenumi}{(FD\arabic{enumi}$_L$)}
\renewcommand{\labelenumi}{\theenumi}

\item \label{it:fdlperiod} 
$\displaystyle{
F_k (\zetav + 2i\pi \ev^{(j)} ) =
 \textcolor{red}{-}\Big(\prod_{\substack{i=1 \\ i \neq j}}^k S(\zeta_i-\zeta_j)\Big)  F_k (\zetav ) .
}$

\item \label{it:fdlrecursion}

The $F_k$ have first order poles at $\zeta_n-\zeta_m = i \pi$, where $1 \leq m < n \leq k$,
and one has with $\hat\zetav = (\zeta_1,\ldots,\widehat{\zeta_m},\ldots, \widehat{\zeta_n}, \ldots, \zeta_k)$ where the hats denote omitted variables,
\begin{equation*}
\res_{\zeta_n-\zeta_m = i \pi} F_{k}(\boldsymbol{\zeta})
= - \frac{1}{2\pi i }
\Big(\prod_{j=m}^{n} S(\zeta_j-\zeta_m) \Big)
\Big(1 + \prod_{p=1}^{k} S(\zeta_m-\zeta_p) \Big)
F_{k-2}( \boldsymbol{\hat\zeta} ).
\end{equation*}

\end{enumerate}
\end{theorem}

Again, this is in line with \cite[Sec.~3]{Lashkevich:1994}. The proof follows the same lines as in Theorem~\ref{thm:ffdcone}, so that we report the changes briefly:
\begin{proof}
  Since $V$ is even by definition, one finds $\cme{m,n}{V}=0$ for $m+n$ odd; hence let $m+n$ be even in the following. In view of Lemma~\ref{lemma:disorder}(\ref{it:vrefl}),(\ref{it:gammalr}), both  
 $U(-r \ev^{(1)})AU(-r\ev^{(1)})^\ast$ and  $\Gamma U(-r\ev^{(1)})U(\refl)AU(\refl)U(-r\ev^{(1)})^\ast$ are local in the wedge $\wedg_L$ and we can apply Theorem~\ref{thm:ffwedge} to them. But these two are linked by the formula
\begin{equation}\label{eq:refldiso}
\cme{m,n}{\Gamma U(\refl)V U(\refl)}(\thetav, \etav) =(-1)^{m} \sum_{D \in \mathcal{C}_{m,n}} \delta_D S_D \hat{R}_D(\thetav, \etav) \cme{n- |D|, m - |D|}{V}(\hat{\etav}, \hat{\thetav}),
\end{equation}
cf.~\eqref{eq:fmnrefl}, where now $\hat R_D$ is the factor $R_D$ from \cite[Eq.~(3.44)]{BostelmannCadamuro:expansion} with the minus sign there replaced by a plus.

The analytic function $F_k$ ($k$ even) can then be defined and extended meromorphically by  
\begin{align}   
  F_k (\zetav) &:=\begin{cases}
    F_k^{[V]}(\zetav) &\text{for }\im\zetav\in\ical^k_+,   
\\
    F_k^{[\Gamma U(\refl)V U(\refl)]}(\zetav +i \piv) &\text{for }\im\zetav\in\ical^k_-,   
                \end{cases}
\\ \label{eq:fdlperiod0}
F_k  (\zetav+2i\piv) &:=  F_k (\zetav),
\\ \label{eq:fdlstairs}
F_k(\zetav)&:= (-1)^{m} F_{k}(\zeta_{m+1},\ldots,\zeta_{k},\zeta_{1}+2i\pi,\ldots,\zeta_{m}+2i\pi)
\end{align}
and use of the tube theorem, where the consistency condition \eqref{eq:refldiso} is used to match the values at common boundaries. Comparing at these boundaries also yields the residue formula \ref{it:fdlrecursion}, while \ref{it:fdlperiod} arises from \eqref{eq:fdlstairs} and \ref{it:fdsymm}.
\end{proof}

\section{Conclusions and Outlook}\label{sec:conclusions}

In this paper, we have shown how fermionic particles can be accommodated in the rigorous construction of integrable models in 1+1 spacetime dimensions; this can be done either by modifying one of the two wedge-local fields, or by deforming the von Neumann algebras associated with wedges (at the level of Borchers triples), or from a different point of view, by deforming the PCT operator. One might say that we have ``fermionized'' these models, although the notion of bosonization and fermionization should be read with some care: As shown in Proposition~\ref{prop:fixpoints}, both the bosonic and the fermionic model can be understood as subnets of a common, non-local net of algebras, but the models are distinct both mathematically (being nonequivalent as graded-local nets) and physically (leading to a different asymptotic particle spectrum and scattering matrix).

We have applied this here to models with one type of scalar massive particle for simplicity, where we have shown that the results are compatible with the description of fermions in the form factor programme, both for local fields and for disorder operators. While it is clear that analogous modifications can be made to more intricate models, in particular with several particle species \cite{LechnerSchuetzenhofer:2012,AL2017}, a more unified description including these cases may be desirable; we hope to return to this point elsewhere.

Thus our methods give rise to a new class of (low-dimensional) interacting quantum field theories that were not previously amenable to a fully rigorous description. One may now investigate how other properties of local observables translate to these situations. For example, quantum energy inequalities known from the bosonic case \cite{BostelmannCadamuroFewster:2013,BostelmannCadamuro:2016} have an analogue in the fermionic case \cite{BCM:qeimulti} although the change in interaction will be reflected in the numerical value of the negative energy bounds.

Finally, we remark that similar deformation schemes should also lead to fermionic analogues of wedge-local models in higher dimensions \cite{GrosseLechner:2007}; in fact, our approach bears some similarity to the models with string- and brane-localized observables described in \cite{BuchholzSummers:2007}.

\section*{Acknowledgements}
The authors would like to thank J.~Mandrysch for pointing out references regarding the treatment of fermions in the form factor programme, and C.~A.~Ferrentino for comments on an earlier version of the draft.
D.C.~is supported by the Deutsche Forschungsgemeinschaft (DFG) within the Emmy Noether grant CA1850/1-1 and within the Research Training Group RTG 2522/1.
H.B.~would like to thank the Institute for Theoretical Physics at the University of Leipzig for hospitality.

\addcontentsline{toc}{section}{References}

\bibliographystyle{alpha} 
\bibliography{../../integrable}

\end{document}